\documentclass[aos,preprint]{imsart}

\RequirePackage[OT1]{fontenc}
\RequirePackage{amsthm,amsmath}
\RequirePackage[authoryear]{natbib}
\RequirePackage[colorlinks,citecolor=blue,urlcolor=blue]{hyperref}
\usepackage{pdfpages}

\arxiv{arXiv:1904.09647}

\startlocaldefs
\numberwithin{equation}{section}
\theoremstyle{plain}
\endlocaldefs

\usepackage{array}
\usepackage{verbatim}
\usepackage{float}
\usepackage{mathrsfs}
\usepackage{multirow}
\usepackage{amssymb}
\usepackage{stmaryrd}
\usepackage{graphicx,subcaption}
\usepackage{xargs}[2008/03/08]
\usepackage{soul}

\usepackage{epsfig}
\usepackage{multirow}
\usepackage{url}
\usepackage{amssymb}
\usepackage{epsfig,natbib}
\usepackage{rotate}
\usepackage{lscape}
\usepackage{subcaption}
\usepackage{float}
\usepackage{mathtools}
\usepackage{booktabs}

\newtheorem{thm}{Theorem}

\newtheorem{prop}{Proposition}
 \newtheorem{example}{Example}

\usepackage{mathrsfs}
\usepackage{enumitem}
\usepackage{ifthen}
\usepackage{MnSymbol}
\usepackage{xcolor}
\usepackage{algorithm,algpseudocode}

\usepackage{mathtools}
\allowdisplaybreaks

\def\F{Fr\'echet}

\def\references{\bibliography{mtv}}
\newcommand*{\myproofname}{Proof}

\global\long\def\expect{\mathbb{E}}
\global\long\def\prob{\mathrm{Pr}}

\global\long\def\real{\mathbb{R}}
\global\long\def\Op{O_{P}}
\global\long\def\manifold{\mathcal{M}}

\global\long\def\asympeq{\asymp}

\global\long\def\diffop{\mathrm{d}}
\newcommandx\tangentspace[2][usedefault, addprefix=\global, 1=\manifold]{T_{#2}#1}

\newcommandx\lpnorm[3][usedefault, addprefix=\global, 1=r, 2=]{\|#3\|_{\mathcal{L}^{#1}}^{#2}}
\newcommandx\lp[1][usedefault, addprefix=\global, 1=p]{\mathcal{L}^{#1}}

\global\long\def\metricspace{\mathcal{M}}
\global\long\def\dist{d}
\global\long\def\TV{\mathrm{TV}}

\global\long\def\metric#1#2{\langle#1,#2\rangle}

\newcommandx\vfnorm[3][usedefault, addprefix=\global, 1=\mu, 2=]{\|#3\|_{#1}^{#2}}
\newcommandx\vfinnerprod[2][usedefault, addprefix=\global, 1=\mu]{\llangle#2\rrangle_{#1}}
\global\long\def\define{:=}

\global\long\def\tdomain{\mathcal{T}}

\newcommandx\opnorm[3][usedefault, addprefix=\global, 1=\mu, 2=]{\vertiii{#3}_{#1}^{#2}}
\newcommandx\fronorm[2][usedefault, addprefix=\global, 1=]{|#2|_{F}^{#1}}
\newcommandx\spd[1][usedefault,addprefix=\global, 1=m]{\mathrm{Sym_{\star}^{+}}(#1)}

\newcommandx\dpower[4][usedefault, addprefix=\global, 1=2, 2=n]{|#3-#4|_{#2}^{#1}}
\newcommandx\dn[3][usedefault, addprefix=\global, 1=n]{\|#2-#3\|_{#1}}
\newcommandx\dnpower[4][usedefault, addprefix=\global, 1=2, 2=n]{\|#3-#4\|_{#2}^{#1}}

\newcommand{\bea}{\begin{eqnarray*}}
\newcommand{\eea}{\end{eqnarray*}}
\newcommand{\be}{\begin{eqnarray}}
\newcommand{\ee}{\end{eqnarray}}
\newcommand{\ed}{\end{document}}

\newcommand{\btab}{\begin{tabular}}
\newcommand{\etab}{\end{tabular}}

\newcommand{\bi}{\begin{itemize}}
\newcommand{\ei}{\end{itemize}}
\newcommand{\bfi}{\begin{figure}}
\newcommand{\efi}{\end{figure}}
\newcommand{\ben}{\begin{enumerate}}
\newcommand{\een}{\end{enumerate}}
\newcommand{\bay}{\begin{array}}
\newcommand{\eay}{\end{array}}

\newcommand{\td}{\tilde{d}}

\def\F{Fr\'echet }

\def\bco{\iffalse}

\def\cp{\citep}

\newcommand{\bc}{\begin{center}}
\newcommand{\ec}{\end{center}}

\newcommand{\linrev}[1]{\textcolor{black}{#1}}

\usepackage{xr}
\begin{document}

\begin{frontmatter}
\title{Total Variation Regularized Fr\'echet Regression for Metric-Space Valued Data}
\runtitle{Total Variation Regularized Fr\'echet Regression}
	
	\begin{aug}
		\author[A]{\fnms{Zhenhua} \snm{Lin}\thanksref{t1}\ead[label=e1]{linz@nus.edu.sg}}
		\and 
		\author[B]{\fnms{Hans-Georg} \snm{M\"uller}\thanksref{t2}\ead[label=e2]{hgmueller@ucdavis.edu}}

		\thankstext{t1}{Research supported by NUS Startup Grant R-155-000-217-133.}
		\thankstext{t2}{Research supported by NSF Grant DMS-1712864.}
		
		\address[A]{Department of Statistics and Applied Probability, National University of Singapore, \printead{e1}}
		
		\address[B]{Department of Statistics, University of California, Davis, \printead{e2}}
	\end{aug}
	
	\begin{abstract}
Non-Euclidean data that are indexed with a scalar predictor such as time are increasingly encountered in data applications, while statistical methodology and theory for such random objects are not well developed yet. To address the need for new methodology in this area,  we develop a total  variation regularization technique for nonparametric Fr\'echet regression, which refers to a regression setting where  a response residing in a metric space is paired with a scalar predictor and the target is a conditional  Fr\'echet  mean. Specifically, we seek to approximate an unknown metric-space valued function by an estimator that minimizes the Fr\'echet version of least squares and at the same time has small total variation, appropriately defined for metric-space valued objects.  We show that the resulting  estimator is representable by a piece-wise constant function and establish the minimax convergence rate of the proposed estimator for metric data objects that reside in Hadamard spaces. We illustrate the numerical performance of the proposed method for both simulated  and real data, including metric spaces of symmetric positive-definite matrices with the affine-invariant distance, of probability distributions on the real line with the Wasserstein distance, and of phylogenetic trees with the Billera--Holmes--Vogtmann metric.
	\end{abstract}
	
	
	
\end{frontmatter}

\section{Introduction}

Regression analysis is a foundational technique in statistics aiming to model
the relationship between response variables and covariates or predictor variables.
Conventional regression models are designed for Euclidean responses $Y$ and predictors $X$ and include parametric models such as 
linear or polynomial regression and generalized linear models 
as well as various nonparametric approaches, such as kernel and spline smoothing. All of these models  target the conditional expectation $\expect(Y|X)$. 

In response to the emergence of new types of data,  the basic  Euclidean regression models 
have been  extended to the case of non-Euclidean data, where a relatively well-studied scenario concerns manifold-valued responses. 
For instance, \citet{Chang1989, Fisher1995} studied regression models for spherical and circular
data, while \citet{Shi2009, Steinke2010a,Davis2010,Fletcher2013,Cornea2017} investigated such models for the case of more general 
Riemannian manifolds. Also classical local regression techniques, such as
Nadaraya--Watson smoothing and local polynomial smoothing,  have been 
generalized to cover responses that lie on manifolds  \cp{Pelletier2006,Yuan2012,Hinkle2014}. 
In this paper, we extend the scope of these previous approaches  and study  the regression problem
for response variables that are situated on a metric space, {more specifically, a Hadamard or Alexandrov space}. Due to the absence of rich geometric and algebraic structure in these metric spaces, this problem
poses new challenges that go beyond the regression problem for the Euclidean or manifold case.

While regression with metric-space valued responses covers  a wide range of random objects and therefore is of intrinsic interest,  the literature on this topic  
so far is quite limited.  Existing work includes \citet{Faraway2014}, who 
considered regression for non-Euclidean data by a Euclidean embedding using distance matrices, similar to multidimensional scaling,  as well as  intrinsic approaches by  \citet{Hein2009}, who studied Nadaraya--Watson kernel
regression for general metric spaces, and by  \citet{Petersen2019}, who introduced 
linear and local linear regression for metric-space valued response
variables and approached the regression problem within the framework of conditional \F means. 

In this paper we propose a novel regularization approach for nonparametric
regression with metric-space valued response variables and a scalar
predictor variable. We  utilize a total variation based penalty, introducing in 
Section \ref{sec:Methodology} an appropriate modification of  the definition of total variation that covers metric-space valued functions.
Specifically,
the inclusion of a total variation penalty
term in the estimating equation for \F regression leads to a penalized  
M-estimation approach for metric-space valued data. 
We refer to
the proposed method as total variation regularized Fr\'echet
regression or simply \textit{regularized Fr\'echet regression}.  While regularized Fr\'echet regression can be developed for any geodesic metric space,  we focus here primarily on the family of Hadamard spaces.  This family includes the Euclidean space and forms a rich class of metric spaces that have  important practical applications; see Examples \ref{exa:spd}--\ref{ex:tree} and Section \ref{sec:application} for more details. 

Total variation regularization was introduced by  \citet{Rudin1992} for image recovery/denoising. There is a vast literature on this regularization technique from the perspective of image denoising and signal processing; see \cite{Chambolle2010} for a brief introduction and review. From a statistical perspective and for Euclidean data, this method was studied by \citet{Mammen1997} from the viewpoint of locally adaptive regression splines and by \cite{Tibshirani2005}, who connected it to  the lasso.  Recent developments along this line include optimal rates \citep{Huetter2016}, trend filtering \citep{Kim2009,Tibshirani2014} and total variation regularized regression when predictors are on a tree or graph \citep{Wang2016,Ortelli2018}. Extensions to manifold-valued
data were first investigated by \cite{penn:06} with a robust variant of the total variation regularization, then by \citet{Lellmann2013,Weinmann2013} with the first-order total variation, and further by \citet{Bergmann2014,Bergmann2016} with the second-order total variation, although  without asymptotic
analysis. Total variation penalties were  also shown to confer advantages for regression models in brain imaging \cp{wang:17}. 
We generalize these approaches  to the case of data in a Hadamard space and
provide a detailed asymptotic analysis for total variation regularized Fr\'echet regression for the first time. While the extension of total variation regularization
from Euclidean spaces to smooth manifolds is  relatively straightforward, 
as one can take advantage of local diffeomorphisms between manifolds
and Euclidean spaces, the generalization to Hadamard spaces, and especially the theoretical analysis, is considerably more challenging. 

We tackle these challenges by leveraging the convexity of the Hadamard space, taking advantage of  the convexity of the distance function and the strong convexity of the squared distance function; see Section \ref{sec:Asymptotics}. Moreover, to overcome the technical difficulties arising from the lack of vector and analytic structures of Hadamard spaces, we develop  new geometric ideas that are relevant for statistical analysis in these spaces, such as Alexandrov inner product, geometric interpolation of metric-space valued functions, and geometric center of functions; see Appendix \ref{sec:o-lem-key} for details. Combined  with convexity, these new constructions enable us to obtain minimax rates of convergence for the proposed estimator for a family of Hadamard spaces and functions of bounded variation. In addition, as these geometric constructions apply to general metric spaces and convexity extends to  certain subspaces of Alexandrov spaces, the theory also applies for certain non-Hadamard spaces.

The structure of  the paper is as follows. A brief introduction to metric geometry is given in Section \ref{sec:metric-geometry}. Total variation regularized \F regression is introduced 
in Section \ref{sec:Methodology}, and asymptotic results are presented in Section \ref{sec:Asymptotics}.  Numerical
studies for synthetic data are provided in Section \ref{sec:Numerical-Studies}. In Section \ref{sec:application} we illustrate the application of the proposed method to analyze data on the  evolution of human mortality profiles using the Wasserstein distance on the space of probability distributions and  to study  the dynamics of brain connectivity using task-related functional magnetic resonance imaging (fMRI) signals and the affine-invariant distance on the space of symmetric positive-definite matrices.

\section{Concepts and Tools from Metric Geometry}
\label{sec:metric-geometry}
To state the  estimation method and theory in Sections \ref{sec:Methodology} and \ref{sec:Asymptotics},  we  need to make use of various concepts from metric geometry that are briefly reviewed here; a more comprehensive treatment can be found in Chapters 2, 4 and 9 of 
\citet{Burago2001} and Chapter VII of \cite{Lang1995}.

\paragraph*{Geodesics}
For  a generic metric space $(\metricspace,\dist)$ and a closed interval  $\tdomain=[a,b]\subset\real$, given  a curve $\gamma$ parameterized by $\tdomain$ on
$\manifold$, i.e., $\gamma:\tdomain\rightarrow\manifold$, and  a set
$P=\{t_{0}\leq t_{1}\leq\cdots\leq t_{k}\}\subset \tdomain$ consisting of $k+1$ points 
in $\tdomain$, we use the quantity 
$R_{d}(\gamma,P)=\sum_{j=1}^{k}d(\gamma(t_{j}),\gamma(t_{j-1}))$ to define the  length of $\gamma$, denoted by $|\gamma|$, which is given by 
\begin{equation} \label{le}
|\gamma|=\sup_{P\in\mathcal{P}}R_{d}(\gamma,P);
\end{equation}
here $\mathcal{P}$ is the collection of subsets of $\tdomain$ whose cardinality
is finite. The metric space $(\manifold,d)$ is a length space
if $d(p,q)=\inf_{\gamma}|\gamma|$, where the infimum ranges over
all curves $\gamma:\tdomain\rightarrow\manifold$ connecting $p$ and $q$,
i.e., $\gamma(a)=p$ and $\gamma(b)=q$. A geodesic on $\manifold$
is a curve $\gamma:\tdomain\rightarrow\manifold$ such that $d(\gamma(s),\gamma(t))=|t-s|$
for $s,t\in \tdomain$. \linrev{The metric space $(\manifold,d)$ is a geodesic space if any pair of points can be connected by a geodesic, and is a uniquely
geodesic space if this geodesic is unique.} The geodesic connecting $p$
and $q$ in a uniquely geodesic space is denoted by $\overline{pq}$. Geodesics in a metric space are the counterpart of straight lines in a Euclidean space. They have been explored for statistical regression of non-Euclidean data, such as geodesic regression \citep{Fletcher2013}.

\paragraph*{Curvature}Unlike Euclidean spaces, a general metric space is often not flat, and curvature is used to measure the amount of deviation from being flat. A standard approach to classifying  curvature is to compare  geodesic triangles
on the metric space to those on the following reference spaces $M_{\kappa}^{2}$:
\begin{itemize}
	\item When $\kappa=0$, $M_{\kappa}^{2}=\real^{2}$ with the standard Euclidean
	distance;
	\item When $\kappa<0$, $M_{\kappa}^{2}$ is the hyperbolic space $\mathbb{H}^{2}=\{(x,y,z)\in\real^3:\,x^2+y^2-z^2=-1\text{ and } z > 0\}$
	with the hyperbolic distance function $d(p,q)=\cosh^{-1}(z_pz_q-x_px_q-y_py_q)/\sqrt{-\kappa}$, where $p=(x_p,y_p,z_p)$ and $q=(x_q,y_q,z_q)$;
	\item When $\kappa>0$, $M_{\kappa}^{2}$ is the sphere $\mathbb{S}^{2}=\{(x,y,z)\in\real^3:\,x^2+y^2+z^2=1\}$ 
	with the angular distance function $d(p,q)=\cos^{-1}(x_px_q+y_py_q+z_pz_q)/\sqrt{\kappa}$.
\end{itemize}
A geodesic triangle with vertices $p,q,r$ in a uniquely geodesic space  $\manifold$, denoted
by $\triangle(p,q,r)$, consists of three geodesic segments that connect
$p$ to $q$, $p$ to $r$ and $q$ to $r$, respectively. A comparison
triangle of $\triangle(p,q,r)$ in the reference space $M_{\kappa}^{2}$
is a geodesic triangle on $M_{\kappa}^{2}$ formed by vertices $\bar{p},\bar{q},\bar{r}$
such that $d(p,q)=\bar{d}_{\kappa}(\bar{p},\bar{q})$, $d(p,r)=\bar{d}_{\kappa}(\bar{p},\bar{r})$,
and $d(q,r)=\bar{d}_{\kappa}(\bar{q},\bar{r})$, where $\bar{d}_{\kappa}$
denotes the distance function on $M_{\kappa}^{2}$. In addition, every point $x$ on the geodesic $\overline{pq}$ ($\overline{pr}$, respectively) has a counterpart $\overline{x}$ on the geodesic segment $\overline{\bar p \bar q}$ ($\overline{\bar p \bar r}$, respectively)  of the comparison triangle such that $d(p,x)=\bar d_{\kappa}(\bar p,\bar x)$.  \linrev{
	We say the (global) curvature of $\manifold$ is lower (upper, respectively) bounded by $\kappa$ if every geodesic triangle with perimeter less than $2D_\kappa$, where $D_\kappa=\pi/\sqrt{\kappa}$ if $\kappa>0$ and $D_\kappa=\infty$ otherwise, satisfies the following property: There exists a comparison triangle $\triangle(\bar p,\bar q,\bar r)$ in $M_{\kappa}$ such that $d(x,y)\geq \bar d_{\kappa}(\bar x,\bar y)$ ($d(x,y)\leq \bar d_{\kappa}(\bar x,\bar y)$, respectively) for all $x\in \overline{pq}$ and $y\in \overline{pr}$ and their comparison points $\bar x$ and $\bar y$ on $\triangle(\bar p,\bar q,\bar r)$.}

\paragraph*{Angles}The comparison angle $\bar{\angle}_{p}(q,r)$ between $q$ and $r$ at
$p$ is defined by 
\begin{equation}\label{eq:comparison-angle}
\bar{\angle}_{p}(q,r)=\arccos\frac{d^{2}(p,q)+d^{2}(p,r)-d^{2}(q,r)}{2d(p,q)d(p,r)}.
\end{equation}
This is utilized to introduce the concept of an (Alexandrov) angle between
two geodesics $\gamma$ and $\eta$ emanating from $p$  in a uniquely geodesic space, which is
denoted by $\angle_{p}(\gamma,\eta)$ and
defined by 
\[
\angle_{p}(\gamma,\eta)=\underset{s,t\rightarrow0}{\lim\sup}\bar{\angle}_{p}(\gamma(s),\eta(t)).
\]
Note that $\angle_{p}(\gamma,\eta)$ does not depend on the length of
$\gamma$ or $\eta$. For three distinct points $p,q,r$ in a uniquely geodesic subset of $\manifold$,
we define the angle $\angle_{p}(q,r)=\angle_{p}(\overline{pq},\overline{pr})$.

\paragraph*{Alexandrov Spaces and Hadamard Spaces} 
\linrev{A geodesic space with lower or upper bounded curvature is called an Alexandrov space, and a complete geodesic space with curvature upper bounded by $0$ is called a Hadamard space. Every geodesic triangle
$\triangle(p,q,r)$  in a Hadamard space then satisfies the $\mathrm{CAT}(0)$
inequality}, i.e., $d(x,y)\leq\bar{d}_0(\bar{x},\bar{y})$ for all $x\in\overline{pq}$ and $y\in\overline{pr}$ and their comparison points $\bar x,\bar y\in\real^2$. 
A geodesic space in which every geodesic triangle satisfies the $\mathrm{CAT}(0)$ inequality is called a CAT(0) space;  a Hadamard space is a complete CAT(0) space. Moreover, every CAT(0) space is uniquely geodesic. Every Euclidean space is a Hadamard space, while non-Euclidean Hadamard spaces include symmetric positive definite matrices, some Wasserstein spaces and phylogenetic tree spaces and more; see Examples \ref{exa:spd}--\ref{ex:tree}. These spaces have broad applications in science and statistics.

\paragraph*{Riemannian Manifolds}
A Riemannian manifold is a smooth manifold with a smooth metric tensor $\metric{\cdot}{\cdot}$,
such that for each $p\in\manifold$, the tensor $\metric{\cdot}{\cdot}_{p}$
defines an inner product on the tangent space $\tangentspace p$ at $p$. The metric tensor induces a distance function that turns the Riemannian manifold into a metric space.
The sectional curvature  at $p$ is defined for two linearly independent tangent vectors $u$ and $v$ at $p\in\manifold$ and is given
by 
$\frac{\metric{\mathfrak R(u,v)v}u_{p}}{\metric uu_{p}\metric vv_{p}-\metric uv_{p}^{2}}\in\real,$
where $\mathfrak R$ is  the Riemannian curvature tensor \citep[p.227,][]{Lang1995}.  
A complete Riemannian manifold is a Hadamard manifold if it is simply connected and has everywhere nonpositive sectional curvature.

\section{Regularized Fr\'echet Regression with Total Variation\label{sec:Methodology}}

Let $(\manifold,d)$ be a {metric space} and $Y$ a random element in 
$\manifold$, where $d$ denotes the distance function on $\manifold$.
When $\manifold$ is a Euclidean space,   which is a special {metric space}, the expectation or mean of $Y$
is an important concept to characterize the average location of $Y$.  For a non-Euclidean
metric space, we replace the mean with the Fr\'echet mean, which is
an element of $\manifold$ that minimizes the Fr\'echet function $F(\cdot)=\expect d^{2}(\cdot,Y)$; in the Euclidean case it 
coincides with the usual mean for random vectors with  finite second moments. In a general metric space with a given probability measure,
 the Fr\'echet mean might not exist, and even when it exists it might
not be unique. We shall assume  that  Fr\'echet means exist and are unique for the random objects we consider in the following. This is the case for  Hadamard spaces when $F(p)<\infty$ for some $p\in\manifold$ \citep{Bhattacharya2003,Sturm2003,Afsari2011,patr:15} and Alexandrov spaces with sufficient concentration assumption and/or additional convexity conditions \citep[Lemma \ref{LEM:FRECHET-MEAN},][]{Lin2021}.

We consider a curve  $\mu:\tdomain\rightarrow\metricspace$ 
on $\metricspace$ that is parameterized by an interval $\tdomain\subset\real$ and that  potentially varies with the sample size $n$. 
Without loss of generality, we assume $\tdomain=[0,1]$
throughout.  {For $n>0$ independent observations $Y_{i}$
at the designated time point $t_i$ for $i=1,\ldots,n$,} we assume the following model \be \label{mun} 
\quad\quad\quad\expect d^2(y,Y_i)<\infty \text{ for some }y\in\manifold, \text{ and }
 \mu(t_{i})=\underset{{y \in\manifold}}{\arg\min}\, \expect d^{2}(y,Y_{i}),
 \ee
 and assume that $0 \le t_{1}\leq\cdots\leq t_{n} \le 1$ are equally spaced; the assumption of equal spacing that we adopt here for simplicity  is not essential, and the results can be easily extended to the non-equally spaced case, by applying the concept of design densities \cp{sack:70}.

Our goal is to obtain 
a mean curve estimate  $\hat{\mu}$ from the  given data  pairs $(t_{i},Y_{i})$ by minimizing 
the  loss function
\[
L_{\lambda}(g)=\frac{1}{n}\sum_{i=1}^{n}d^{2}\left(g(t_{i}),Y_{i}\right)+\lambda\TV(g),
\]
where $\TV(g)=|g|$ is the total variation of the curve $g$,
measured by its length as defined by eq. (\ref{le}),  and  $\lambda\geq0$ is  a regularization parameter depending on $n$. 
The curve estimate is then 
\begin{equation}
\hat{\mu}\in\underset{\TV(g)<\infty}{\arg\min}\,L_{\lambda}(g),\label{eq:definition-regularized-estimator}
\end{equation}
and its deviation from the target $\mu$ is quantified by the pseudo-metrics \be d_{n}(\hat\mu,\mu)=\left\{n^{-1}\sum_{i=1}^{n}\dist^{2}(\hat\mu(t_{i}),\mu(t_{i}))\right\}^{1/2} \label{dn}, \ee where  $\td$ is a pseudo-metric if  $\td(f,g)=\td(g,f)\geq 0$ and $\td(f,h)\leq \td(f,g)+\td(g,h)$ for all $f,g,h$. In the above, both $L_\lambda$ and $d_n$ are empirical, in the sense that they compare $g$ and $\hat\mu$ with their respective targets only at the design points  $t_1,\ldots,t_n$. Nevertheless, the theory developed in the next section implies that with probability tending to one $\hat\mu$ converges to $\mu$, in the sense that $\int_{\tdomain} d^2(\hat\mu(t),\mu(t))\diffop t\rightarrow0$, under  the assumption $\TV(\mu)\leq C$ for a fixed constant $C\geq 0$ and a suitable asymptotic assumption on the spacing of the design points $t_i$ that will be satisfied for example if these points are equidistantly distributed over an interval. 

The estimator $\hat\mu$, although not unique, has the property that  $\hat\mu(t)=\hat\mu(t_1)$ for $t\in [0,t_1]$ and $\hat\mu(t)=\hat\mu(t_n)$ for $t\in[t_n,1]$. Otherwise, the following function 
	$$
	\check\mu(t)=
	\begin{cases}
	\hat\mu(t_1) & \text{for } t\in[0,t_1),\\
	\hat\mu(t) & \text{for } t\in[t_1,t_n],\\
	\hat\mu(t_n) & \text{for } t\in(t_n,1]
	\end{cases}
	$$
	satisfies $n^{-1}\sum_{i=1}^n d^2(\check\mu(t_i),Y_i)=n^{-1}\sum_{i=1}^n d^2(\hat\mu(t_i),Y_i)$ and  $\mathrm{TV}(\check\mu)<\mathrm{TV}(\hat\mu)$, which implies $L_\lambda(\check\mu)<L_\lambda(\hat\mu)$ and thus contradicts  the optimality of $\hat\mu$.
Indeed, the following result shows that $\hat\mu$ can be chosen to have a simple structure.
\begin{prop}
\label{prop:characterization-estimator}For any $\tilde{\mu}$ that
minimizes $L_{\lambda}(\cdot)$, there is a step function $\hat{\mu}$
such that $\hat{\mu}(t_{i})=\tilde{\mu}(t_{i})$ for all $i=1,\ldots,n$
and $\TV(\hat{\mu})\leq\TV(\tilde{\mu})$.
\end{prop}
\begin{proof}
It is clear that $\TV(\tilde{\mu})\geq\sum_{i=0}^{n}d(\tilde{\mu}(t_{i+1}),\tilde{\mu}(t_{i}))$, where $t_{0}=0$ and $t_{n+1}=1$.
Define 
\[
\hat{\mu}(t)=\begin{cases}
\tilde{\mu}(t_{i}),  & t\in[0,1)\text{ and }t\in[t_{i},t_{i+1}),\\
\tilde{\mu}(t_{n}),  & t=1.
\end{cases}
\]
Then $\hat{\mu}(t_{i})=\tilde{\mu}(t_{i})$ for $i=1,\ldots,n$. Also,
from the definition, $\hat{\mu}(t)$ is constant over $[t_{i},t_{i+1})$.
One thus finds  $\TV(\hat{\mu})=\sum_{i=0}^{n}d(\hat{\mu}(t_{i+1}),\hat{\mu}(t_{i}))=\sum_{i=0}^{n}d(\tilde{\mu}(t_{i+1}),\tilde{\mu}(t_{i}))\leq\TV(\tilde{\mu})$.
\end{proof}
The above proposition shows that one  can always choose a step function
to minimize the loss function $L_{\lambda}$. In the following, we may therefore 
assume that $\hat{\mu}$ is a step function. The class of step functions is not only sufficiently
powerful to approximate any function of finite total variation, but
also advantageous in modeling functions that are discontinuous since
it incorporates jumps of the function estimates, in contrast to classical smoothing 
methods that usually assume a smooth underlying regression function. 
Incorporating jumps or discontinuities is  of interest in many applications \cp{kola:12,zhu:14:1,mull:20:2}. Our approach makes it possible to go beyond Euclidean spaces and to fit metric-space valued functions with jumps, as demonstrated in Section \ref{sec:functional-connectivity}.

The tuning parameter $\lambda$ controls the number of constant pieces of the estimate $\hat\mu$ and the magnitude of the distance between the pieces. For instance, a large value of $\lambda$ leads to a small number of constant pieces. In the next section we will show that the choice $\lambda \asympeq n^{-2/3}$ will  optimize the asymptotic performance, where the notation $\lambda\asympeq n^{-2/3}$ denotes that there are constants $c_2\geq c_1>0$ such that $c_1n^{-2/3}\leq \lambda \leq c_2n^{-2/3}$. 
In practice, $\lambda$ can be chosen via cross-validation. In  some situations it is  useful to  choose it as the minimal number that yields a desired number of pieces of $\hat{\mu}$; see Section \ref{sec:functional-connectivity}. For computation of $\hat\mu$, we adopt the iterative proximal point algorithm  of  \cite{Weinmann2013}, who  showed that this algorithm is convergent for Hadamard spaces; further details are  in Appendix \ref{sec:computation}.

\section{Theory\label{sec:Asymptotics}}

\subsection{Hadamard Manifolds and Spaces}\label{subsec:main-theory}
To study the asymptotic properties of the estimate $\hat{\mu}$ given in \eqref{eq:definition-regularized-estimator}, we assume uniform sub-Gaussianity of the random quantities $d(\mu(t_i),Y_i)$, as follows.  A random variable $X$ is  sub-Gaussian if $\expect \exp(\beta X^2)<\infty$ for a constant $\beta>0,$ and a   collection $\mathcal{X}$ of random variables is uniformly sub-Gaussian, if there are constants $\beta,\zeta>0$ such that $\expect \exp(\beta X^2)\leq \zeta<\infty$ for all $X\in\mathcal{X}$. The following  condition states that the distances of random objects $Y_{i}$ to their Fr\'echet means are  uniformly sub-Gaussian. This is guaranteed and thus the condition is not needed  whenever the diameter of the space $\manifold$ is bounded.  
{\begin{enumerate}[label=(\textbf{H\arabic*})]
		\item\label{cond:H:subG} There exist constants $\beta>0$ and $\zeta>0$ such that for the data $Y_i$  in  model \eqref{mun} \begin{equation*}\label{eq:sub-gaussian}\sup_{1\leq i\leq n}\expect[\exp\{\beta d^2(\mu(t_{i}),Y_{i}) \}]\leq \zeta<\infty,\end{equation*}
		i.e., the random variables  $d(\mu(t_{i}),Y_{i})$
		are uniformly sub-Gaussian.
\end{enumerate}}

 Let  $\mathscr{V}_{\manifold}$ be the collection of all $\manifold$-valued curves of bounded total variation. We focus on a subcollection $\mathscr{G}_{\manifold}\subset \mathscr{V}_{\manifold}$, which could correspond to  the entire collection $\mathscr{V}_{\manifold}$ or a proper subcollection of $\mathscr{V}_{\manifold}$ such as the class of Lipschitz continuous curves. Then  
the pseudo-metric function $d_n$ in \eqref{dn}  turns $\mathscr G_{\manifold}$ into a pseudo-metric space. 
Let $\mathscr G_{\manifold}^R(C)\subset \mathscr G_{\manifold}$ be a collection of functions $g\in\mathscr G_{\manifold}$ with $\TV(g)\leq C$, such that there exists a ball $\mathcal B\subset \manifold$ of radius $R>0$ with $g(t)\in \mathcal B$ for all $g$ and $t$; we write $\mathscr G_{\manifold}(C)=\mathscr G^\infty_{\manifold}(C)$. The following result, valid for any (non-unique) minimizer $\hat\mu$ in  \eqref{eq:definition-regularized-estimator}, establishes the convergence rate of the estimator $\hat\mu$ for $\mu$, where $\mu$ is allowed to vary with the sample size $n$.

\begin{thm}
	\label{thm:hadamard-manifold} For a family $\mathscr R(p,\kappa)$ of complete and simply connected Riemannian manifolds of   dimension no larger than $p$ and with sectional curvature bounded between $\kappa\leq 0$ and $0$, choosing  $\lambda\asympeq n^{-2/3}$ implies that
	$$\underset{D\rightarrow\infty}{\lim}\underset{n\rightarrow\infty}{\lim\sup}\,\underset{F\in\mathscr{F}_n}{\sup}\,\mathbb{P}_F\{d_{n}(\hat{\mu},\mu)>Dn^{-1/3}\}=0,$$
	where $\mu$ is defined in \eqref{mun}, $\hat\mu$ is given in \eqref{eq:definition-regularized-estimator}, $\mathbb P_F$ is the probability measure 
	induced by $F$, and $\mathscr{F}_n=\mathscr{F}_n(p,\kappa,C,\beta,\zeta)$ for constants $p,C,\beta,\zeta>0$  and $\kappa \le 0$ is the collection of joint probability distributions of $Y_1,\ldots,Y_n$ on $\manifold$ for which   $\manifold\in\mathscr R(p,\kappa)$, $\TV(\mu)\leq C$ and \ref{cond:H:subG} holds for $\beta, \zeta>0$.
\end{thm}

The manifold in the above theorem is a Hadamard manifold which is also a Hadamard space according to Theorem 1A.6 of \cite{Bridson1999}. This motivates us to  generalize the above result to general Hadamard spaces that are not a manifold. To this end, we first observe that  Riemannian manifold-valued functions of bounded total variation satisfy an entropy condition, as follows. 
For a subset $\mathscr B$ of $\mathscr G_{\manifold}$, the minimal number of balls of radius $\delta$ in $(\mathscr G_{\manifold},d_n)$ to cover $\mathscr B$ is denoted by  $N(\delta,\mathscr{B},d_{n})$. 
The covering number $N(\delta,\mathscr{B},d_{n})$  depends on $d_n$, which in turn depends on the metric $d$ as per (\ref{dn}).  Proposition \ref{prop:entropy-tv-ball} in Appendix \ref{sec:aux}  shows that manifolds $\manifold$ in the family $\mathscr R(p,\kappa)$ of Theorem \ref{thm:hadamard-manifold} satisfy the following condition.
\begin{enumerate}[label=(\textbf{H\arabic*})]
	\setcounter{enumi}{1}
	\item\label{cond:H:entropy}For a fixed $R>0$, there exists a constant $K>0$ that  may depend on $R$, such that  $\log N(\delta,\mathscr G_{\manifold}^r (r),d_{n})\leq K\delta^{-1}$ for all $\delta>0$, $n\geq 1$ and $0<r\leq R$.
\end{enumerate}
This condition essentially controls the (local) complexity of the underlying space $\manifold$, {and is key for the asymptotic analysis based on empirical process theory, such as \cite{Mammen1997}.} 
For those Hadamard spaces  and classes $\mathscr G_{\manifold}$ of functions that satisfy the condition, we have the following result that generalizes Theorem \ref{thm:hadamard-manifold}.
\begin{thm}
	\label{thm:hadamard}For $C>0$, for a family $\mathscr H(K)$ of Hadamard spaces such that for each $\manifold\in\mathscr H(K)$ the class of functions $\mathscr G_{\manifold}$ satisfies the condition \ref{cond:H:entropy} for $R=15C$, with $\lambda\asympeq n^{-2/3}$,
	one has $$\underset{D\rightarrow\infty}{\lim}\underset{n\rightarrow\infty}{\lim\sup}\,\underset{F\in\mathscr{F}_n}{\sup}\,\mathbb{P}_F\{d_{n}(\hat{\mu},\mu)>Dn^{-1/3}\}=0,$$
	where $\mu$ is defined in \eqref{mun}, $\hat\mu$ is given in \eqref{eq:definition-regularized-estimator}, $\mathbb P_F$ is the probability measure 
	induced by $F$, and $\mathscr{F}_n=\mathscr{F}_n(K,C,\beta,\zeta)$ for constants $K,C,\beta,\zeta>0$ is the collection of joint probability distributions of $Y_1,\ldots,Y_n$ on $\manifold$ for which   $\manifold\in\mathscr H(K)$, $\mu\in\mathscr G_{\manifold}(C)$, and \ref{cond:H:subG} holds for $\beta, \zeta>0$.
\end{thm}
When $\manifold$ is the one-dimensional Euclidean space $\real$, \cite{Donoho1998} showed that the minimax rate is $n^{-1/3}$ for the class of uniformly bounded variation; see also \cite{Sadhanala2016}. Since $\mathscr H(K)$ contains the one-dimensional Euclidean space  for the same class of functions, the rate in the above theorem is also the minimax rate for the family $\mathscr H(K)$; our result is thus a generalization of the minimax result of \cite{Donoho1998} to Hadamard spaces. In addition, if the entropy condition of \ref{cond:H:entropy} is replaced with $\log N(\delta,\mathscr G_{\manifold}^r (r),d_{n})\leq K\delta^{-\alpha}$ for some constant $\alpha\in(0,2)$, then the proof of Theorem \ref{thm:hadamard} can be modified to show that $d(\hat\mu,\mu)=\Op(n^{-1/(2+\alpha)})$.

There are various geometric properties of Hadamard spaces  that enable the extension in Theorem \ref{thm:hadamard};  the most important among these is the convexity outlined in the following proposition.
\begin{prop}\label{LEM:KEY}
	For $C>0$, let $\mathscr M(K)$ be a family of metric spaces such that for each $\manifold\in\mathscr M(K)$ the class $\mathscr G_{\manifold}(C)$ of   functions satisfies \ref{cond:H:entropy} with $R=15C$. In addition, the following conditions hold  for a universal constant $C_1>0$: For each $\manifold\in\mathscr M(K)$,
	\begin{enumerate}[label=\textup{(\alph*)}]
		\item\label{lem:key:2} $d^{2}(q,r)\geq d^{2}(p,r)-2d(p,q)d(p,r)\cos\angle_{p}(q,r) +d^{2}(p,q)$ 
		for all $p,q,r\in\manifold$;
		\item\label{lem:key:3} the function $f(r)=d(p,r)\cos\angle_{p}(q,r)$ is Lipschitz continuous with a Lipschitz constant no larger than $C_1$ for all $p,q\in\manifold$;
	\item\label{lem:key:4} $\expect \{d(\mu(t_i),Y_i)\cos\angle_{\mu(t_i)}(Y_i,q)\}\leq 0$
	for all $q\in\manifold$, $n\geq1$ and $1\leq i\leq n$.
\end{enumerate}
	For $\lambda\asympeq n^{-2/3}$, it then holds that
	\be \label{prop2} \underset{D\rightarrow\infty}{\lim}\underset{n\rightarrow\infty}{\lim\sup}\,\underset{F\in\mathscr{F}_n}{\sup}\,\mathbb{P}_F\{d_{n}(\hat{\mu},\mu)>Dn^{-1/3}\}=0,\ee
	where $\mu$ is defined in \eqref{mun}, $\hat\mu$ is given in \eqref{eq:definition-regularized-estimator}, $\mathbb P_F$ is the probability measure induced by $F$, and $\mathscr{F}_n=\mathscr{F}_n(K,C,\beta,\zeta)$ is a collection of joint probability distributions  of $Y_1,\ldots,Y_n$ on $\manifold\in\mathscr M(K)$ such that $\mu\in\mathscr G_{\manifold}(C)$, and the conditions (c) and \ref{cond:H:subG} hold for $\beta,\zeta>0$.
\end{prop}
The first two conditions of the above proposition emerge as  properties of Hadamard space. 
In fact, condition \ref{lem:key:2} is an  alternative characterization of the CAT(0) space, which has nonpositive curvature (also known as NPC space). To see this, by Proposition 1.7 in Chapter II.1 of \cite{Bridson1999}, $\manifold$ is a CAT(0) space if and only if for all $p,q,r\in\manifold$, $d(q,r)\geq \bar d_0(\bar q,\bar r)$, where $\bar p,\bar q,\bar r$ form a triangle in the reference space $M_0^2=\real^2$ such that $d(p,q)=\bar d_0(\bar p,\bar q)$, $d(p,r)=\bar d_0(\bar p,\bar r)$ and $\angle_{\bar p}(\bar q,\bar r)=\angle_p(q,r)$. Then, by the law of cosines one further has $d^2(q,r)\geq\bar d_0^2(\bar q,\bar r)=\bar d_0^2(\bar p,\bar q)-2\bar d_0(\bar p,\bar q)\bar d_0(\bar p,\bar r)\cos\angle_{\bar p}(\bar q,\bar r)+\bar d_0^2(\bar p,\bar r)=d^2(p,q)-2d(p, q) d(p,r)\cos\angle_{p}(q,r)+d^2(p,r)$.As the condition \ref{lem:key:2} implies that $\manifold$ is a CAT(0) space which is uniquely geodesic, the angles $\angle_p(q,r)$ and $\angle_{\mu(t_i)}(Y_i,q)$ in Proposition \ref{LEM:KEY} are well defined. 
Verification of the Lipschitz condition \ref{lem:key:3} is nontrivial for a general Hadamard space. Using  various properties of the Hadamard space, we show in Lemma \ref{lem:lipschitz-proj} \citep{Lin2021} that condition \ref{lem:key:3}  holds for all Hadamard spaces with the universal constant $C_{1}=5$. Finally, Lemma \ref{LEM:FRECHET-MEAN} \citep{Lin2021} shows that condition \ref{lem:key:4} also holds for Hadamard spaces. Consequently, Theorem \ref{thm:hadamard} follows directly from Proposition \ref{LEM:KEY}, and Theorem \ref{thm:hadamard-manifold} follows as a special case of Theorem \ref{thm:hadamard}.

The CAT(0) inequality, which holds for Hadamard spaces, implies the convexity of the distance function, i.e., \begin{equation}\label{eq:lem:key:1}d(\llbracket p,q\rrbracket_{\theta},\llbracket p,r\rrbracket_{\theta})\leq \theta d(q,r) \text{ for all }\theta\in [0,1] \text{ and all }p,q,r\in\manifold,\end{equation} 
where  $\llbracket p,q\rrbracket_{\theta}$ denotes  the point that sits  on the geodesic segment connecting $p$ to $q$ and satisfies  $d(p,\llbracket p,q\rrbracket_{\theta})=\theta d(p,q)$.
This convexity is used to bound the total variation of the geodesically interpolated functions $\tilde{g}_\theta(t)=\llbracket \mu(t),g(t)\rrbracket_{\theta}$ by the total variation of the functions $\mu$ and $g$; see Section \ref{sec:pf-lem-key} of the supplementary article  \citep{Lin2021}. We provide an overview of the main steps of the proof of Proposition~\ref{LEM:KEY} demonstrating  how it relies on new geometric ideas that are introduced here to establish this key result  in  Appendix~\ref{sec:o-lem-key}, while the detailed steps of the proof are provided  in Section~\ref{sec:pf-lem-key} of the supplementary materials.

In the following, we discuss three pertinent examples which  will also be further investigated in simulations and data applications.
\begin{example}[{\it Symmetric positive-definite matrices}] 
	\label{exa:spd}{\rm Symmetric positive-definite (SPD) matrices as random objects arise
		in many applications that include  computer vision \citep{Rathi2007}, 
		medical imaging \citep{Fillard2005,Arsigny2006,penn:06,Fletcher2007,Dryden2009} and neuroscience
		\citep{Friston2011}. For example, diffusion tensor imaging, which is
		commonly used to obtain brain connectivity maps based on magnetic resonance imaging (MRI),
		produces $3\times3$ SPD matrices that characterize the local
		diffusion \cp{zhou:16}. For the space of $m\times m$ SPD matrices,
		denoted by $\spd$, the Euclidean distance function $d_{E}(A,B)=\|A-B\|_{F}$ that is based on the 
		Frobenius norm  $\|\cdot\|_{F}$ suffers
		from the so-called swelling effect: The determinant of the average SPD matrix
		is larger than any of the individual determinants \citep{Arsigny2007}. Rectifying this issue 
		motivates the use of  more sophisticated distance functions,  such as
		the Log-Euclidean distance $d_{LE}(A,B)=\|\log A-\log B\|_{F}$ \citep{Arsigny2007},  
		the affine-invariant distance $d_{AI}(A,B)=\|\log(A^{-1/2}BA^{-1/2})\|_{F}$
		\citep{Moakher2005,penn:06} or the Log-Cholesky distance \citep{Lin2019a}, 
		where $\log A$ is the matrix logarithm of $A$. Either of the above distance functions is indeed induced by a Riemannian metric tensor that turns $\spd$ into a complete and simply connected Riemannian manifold of nonpositive and bounded sectional curvature. Therefore, Theorem~\ref{thm:hadamard-manifold} applies to this case.} 
\end{example}

\begin{example}[{\it Wasserstein space} $\mathcal{W}_{2}(\real)$]
\label{exa:wass} {\rm Let $\mathcal{W}_{2}(\real)$ be the space of probability
distributions on the real line $\real$, equipped with the  Wasserstein distance
$d_{W}(G_{1},G_{2})=[\int_{0}^{1}\{G_{1}^{-1}(s)-G_{2}^{-1}(s)\}^{2}\diffop s]^{1/2}$,  where $G_{1}^{-1}$
and $G_{2}^{-1}$ are the (left continuous) quantile functions
corresponding to distribution functions $G_{1}$ and $G_{2}$. According
to Proposition 4.1 of \citet{Kloeckner2010}, $\mathcal{W}_{2}(\real)$
is a CAT(0) space. As $\mathcal{W}_{2}(\real)$ inherits the completeness
of $\real$, $\mathcal{W}_{2}(\real)$ is also a Hadamard space. We illustrate the utility of $\mathcal{W}_{2}(\real)$ for data analysis in  a
study of mortality profiles in  Section \ref{sec:mortality}.  
As in the proof of Proposition 1 of \cite{Petersen2019}, one can show that $\sup_{G\in \mathcal W_2(\real)} \log N(\epsilon\delta,B_G(\delta),d_W)\leq K\epsilon^{-1}$ for a constant $K$ and all $\delta,\epsilon>0$, where $B_G(\delta)=\{\tilde G\in \mathcal W_2(\real):d_W(G,\tilde G))\leq \delta\}$. 
Then, for the function class 
$\mathscr G$ of Lipschitz continuous $\mathcal{W}_{2}(\real)$-valued functions defined on $\tdomain$, using Proposition \ref{prop:entropy} in Appendix \ref{sec:aux}, we can establish condition \ref{cond:H:entropy}, and therefore the rate in Theorem~\ref{thm:hadamard} applies. It is worth noting that $\mathcal{W}_{2}(\real^m)$ is not a Hadamard space for $m\geq 2$ \cite[Section 4,][]{Kloeckner2010}, so that Theorem \ref{thm:hadamard} does not apply.}
\end{example}

\begin{example}[{\it Phylogenetic trees}]\label{ex:tree}
		{\rm Phylogenetic trees are central  data objects in the field of evolutionary biology, where they are used to represent the evolutionary history of a set of organisms. In a  seminal paper by \cite{Billera2001}, phylogenetic trees with $m$ leaves are modeled by metric  $m$-trees endowed with a metric that turns the space of phylogenetic $m$-trees into a metric space, as follows. A leaf is a vertex that is connected by only one edge, and a metric $m$-tree is a tree with $m$ uniquely labeled leaves and positive lengths on all interior edges, where an edge is called an interior edge if it does not connect to a leaf. A collection of $m$-trees that have the same tree structure (taking  leaf labels into account) but different edge lengths can be identified with the orthant $(0,\infty)^r$, where $r$ (determined by the tree structure) is the number of interior edges of each tree in the collection. Collections of different tree structures, identified by different orthants, can be glued together along the common faces of the orthants. With this identification between points and metric $m$-trees, a natural distance function $d_T$ on the space $\mathscr T_m$ of all metric $m$-trees is defined in the following way: For  two trees in the same orthant, their distance is the Euclidean distance, while for two trees from different orthants, their distance is the minimum length over all paths that connect them and consist of only connected segments, where a segment is a straight line within an orthant. 
		According to Lemma 4.1 of \cite{Billera2001}, the space ($\mathscr{T}_m,d_T)$ is a CAT(0) space. In addition, as a cubical complex, by Theorem 1.1 of \cite{Bridson1991} 
		it is also a complete metric space and thus a Hadamard space. For a fixed $m$, from the construction of $\mathscr T_m$, one can see that the covering number $N(\epsilon\delta,B_x(\delta),d_T)$ for the ball $B_x(\delta)$  centered at $x\in\mathscr T_m$ and with radius $\delta$ is of the same order as the covering number of the unit ball of a finite-dimensional  Euclidean space, which is $O(\epsilon^{-k})$ for a $k=k(m) \ge 1$.  For the function class $\mathscr G$  of $\mathscr{T}_m$-valued Lipschitz continuous functions, using Proposition \ref{prop:entropy} in Appendix \ref{sec:aux}, one finds that  the condition \ref{cond:H:entropy} holds for $\mathscr T_m$ and $\mathscr G$. Therefore, Theorem~\ref{thm:hadamard} applies to this case.}
\end{example}

\subsection{Extension to Alexandrov Spaces}\label{sec:alexandrov}
The development of our main results crucially depends on the convexity of the Hadamard space, characterized by  condition \ref{lem:key:2} of Proposition~\ref{LEM:KEY},  which is shown to be equivalent to the CAT(0) inequality and implies  the convexity \eqref{eq:lem:key:1} of the distance function of the Hadamard space. By examining the proofs of Proposition~\ref{LEM:KEY} and Lemma~\ref{lem:key:rate-with-logn} in the supplementary article  \citep{Lin2021}, one finds that  condition \ref{lem:key:2} can be relaxed to \begin{equation}\label{eq:relaxed-cosine-law}d^{2}(q,r) \geq d^2(p,q) - 2d(p,r)d(p,q)\cos \angle_p(q,r)+cd^{2}(p,r)\end{equation} for a universal constant $c>0$, where we note that $c=1$ for Hadamard spaces. 
It turns out that inequality \eqref{eq:relaxed-cosine-law} holds for some subspaces of Alexandrov spaces with positive lower and upper bounded curvature, and thus our main results potentially carry over to such subspaces.

Another key ingredient is the strong convexity of the squared distance function of a  Hadamard space.  A real-valued function $f$ defined on a convex subset of $\real^k$ is strongly convex with parameter $\eta>0$ if $f((1-\theta)p+\theta q)\leq (1-\theta)f(p) + \theta f(q) -\eta \theta(1-\theta) \|p-q\|^2$ for all $p,q$ in the convex subset and $\theta\in[0,1]$. To generalize this concept to functions with  geodesic-metric-space valued arguments, we observe that the convex combination $(1-\theta)u+\theta v$ lies on the straight line connecting $u$ and $v$, and is conveniently  replaced with a point on the geodesic connecting $p$ and $q$. Specifically, we refer to  a function $f$ defined on a geodesically  convex subset 
	$\mathcal{C}$ of a geodesic space as a strongly convex function on $\mathcal{C}$ with parameter $\eta>0$ if $f(\llbracket p,q\rrbracket_\theta)\leq (1-\theta)f(p)+\theta f(q)-\eta\theta(1-\theta)d^2(p,q)$ for all $p,q\in\mathcal C$ and $\theta\in[0,1]$, where a subset in a geodesic space is geodesically convex if for any two points in the subset there exists a unique geodesic contained within the subset that connects those two points. One of the nice properties of strongly convex functions is the existence and uniqueness of a minimizer on  a geodesically convex closed subspace when the function is continuous \citep[Proposition 1.7,][]{Sturm2003}. For any fixed element $q$  of a Hadamard space, the function $f(\cdot)= d^2(\cdot,q)$ that is defined on this space  is continuous and strongly convex with parameter $\eta=1$ \citep[Eq (2),][]{Bacak2015}. This implies the strong convexity of the Fr\'echet function $F(\cdot)=\expect d^2(\cdot,Y)$, whence the Fr\'echet mean of a random object on a Hadamard space always exists and is unique provided that the Fr\'echet function is finite. For specific Alexandrov spaces, the squared distance function shares the property of being strongly convex over some geodesically convex subspaces; see Example \ref{exa:sphere} below.

Utilizing   strong convexity and the relaxed  condition \eqref{eq:relaxed-cosine-law} makes it possible to extend the main results in Section \ref{subsec:main-theory} to certain Alexandrov spaces. Let $\manifold$ be an Alexandrov space with positive lower and upper bound on curvature, where the upper bound is denoted by $\kappa$. The space $\manifold$ generally has a 
finite diameter, according to Theorem 1.9 of \citet{Petrunin1999}. Consequently, the sub-Gaussianity condition \ref{eq:sub-gaussian} is automatically satisfied for all random objects in $\manifold$. We need the following additional assumptions.
\begin{enumerate}[label=(\textbf{A\arabic*})]
	\item\label{cond:A1}  There exists $Q>0$ 
	such that $\log N(\delta,\mathscr G_{\manifold}^r (r),d_{n})\leq rQ\delta^{-1}$
	for all $\delta>0$ and $r>0$.  
	\item\label{cond:A2} There exists a geodesically convex closed subset
	$\mathcal{C}\subset\manifold$ of diameter less than $\pi/(2\sqrt{\kappa})$ such that 
	\begin{enumerate}[label=(A2\alph*)]
		\item\label{cond:A2:a} $Y_i\in\mathcal{C}$ for all $n\geq1$ and $1\leq i\leq n$, 
		\item\label{cond:A2:b} the function $h(x)=d^2(x,y)$ is strongly convex with a universal constant $C_2>0$ for all $y\in\mathcal{C}$, and
		\item\label{cond:A2:c} $d^{2}(q,r) \geq d^2(p,q) - 2d(p,r)d(p,q)\cos \angle_p(q,r)+C_3d^{2}(p,r)\geq0$ for a universal constant $C_3>0$ and all $p,q,r\in\mathcal{C}$.
	\end{enumerate}
\end{enumerate} 
The entropy condition \ref{cond:A1} is a simplified version of the condition \ref{cond:H:entropy}, as now the space $\manifold$ is of bounded diameter.  
The bound on the diameter of the subset $\mathcal C$ implies that $\mathcal C$ is a uniquely geodesic subset of $\manifold$ and thus ensures that the angle $\angle_p(q,r)$ in \ref{cond:A2:c} is well defined. 
 As previously mentioned, the strong convexity condition \ref{cond:A2:b} implies the existence and uniqueness of the Fr\'echet mean,  
 and  \ref{cond:A2:c} is a relaxation of  condition \ref{lem:key:2} of Proposition~\ref{LEM:KEY}. Then, with an argument similar to the proof of Proposition~\ref{LEM:KEY}, the following holds. 
\begin{thm}
	\label{THM:SPHERE}For a family $\mathscr A(R,\kappa)$ of positively curved Alexandrov spaces, all of which have a diameter bounded by $R$ and a  curvature upper bounded by $\kappa>0$, with $\lambda\asympeq n^{-2/3}$,
	one has  $$\underset{D\rightarrow\infty}{\lim}\underset{n\rightarrow\infty}{\lim\sup}\,\underset{F\in\mathscr{F}_n}{\sup}\,\mathbb{P}_F\{d_{n}(\hat{\mu},\mu)>Dn^{-1/3}\}=0,$$
	where $\mu$ is defined in \eqref{mun}, $\hat\mu$ is given in \eqref{eq:definition-regularized-estimator}, $\mathbb P_F$ is the probability measure 
	induced by a probability distribution $F$, and $\mathscr{F}_n=\mathscr{F}_n(R,\kappa,Q,C,C_2,C_3)$ for constants $R,\kappa,Q,C,C_2,C_3$ is the collection of joint probability distributions of $Y_1,\ldots,Y_n$ on $\manifold$ for which   $\manifold\in\mathscr A(R,\kappa)$, $\mu\in\mathscr G_{\manifold}(C)$, and conditions \ref{cond:A1}--\ref{cond:A2} hold.
\end{thm}

\begin{example}[{\it Time-indexed compositional data}]\label{exa:sphere}  {\rm Such data arise in various settings that include  longitudinal compositional data
	\citep{dai:17:1}.  Specifically, for compositional data 
	$Y_i=(z_{i,1},\ldots,z_{i,k+1})$ such that $z_{i,j}\geq0$ and $\sum_{j=1}^{k+1}z_{i,j}=1,\,\,\, i=1, \ldots,n$,
	one may  apply the square root transformation on each $z_{i,j}$ and view
	$(\sqrt{z_{i,1}},\ldots,\sqrt{z_{i,k+1}})$ as elements of the  quadrant
	$\mathcal{C}=\{(x_{1},\ldots,x_{k+1})\in\mathbb{S}^{k}:x_{j}\geq0\text{ for }j=1,\ldots,k+1\}$.
	Compositional data can thus be viewed as sampled from
	the convex subset $\mathcal{C}$, where the diameter of this quadrant
	is $\pi/2$. Then, for all $p,q,r\in\mathcal{C}$, whenever $d(q,r)\leq c_{1}<\pi/2$
	for a universal constant $c_{1}>0$,  according to the Taylor expansion of the function $h(\cdot)=d^{2}(\cdot,r)$ at $p$
	and its gradient and Hessian \citep[Supplement A,][]{Pennec2018}, we find that \eqref{eq:relaxed-cosine-law} holds for some universal constant $c=c_2>0$ (depending on $c_{1}$). In addition, the Hessian of $h$  is positive on $\mathcal C$ uniformly for all $r\in\mathcal{C}$, {which implies the strong convexity of $h$.} 
	Then condition \ref{cond:A2} is satisfied if $\prob\{d(\partial\mathcal{C},Y_{i})\geq c_{3}\text{ for all }1\leq i\leq n\}=1$
	for a universal constant $c_{3}>0$, where $\partial\mathcal{C}$
	is the boundary of $\mathcal{C}$ and $d(\partial\mathcal{C},p)$
	is the distance of $p$ to the set $\partial\mathcal{C}$. This requirement 
	 corresponds to points being 
	not too close to the boundary of $\mathcal{C}$. This  is a 
	 mild condition, as $c_{3}$ can be
	arbitrarily small.  For the  class $\mathscr G$ of $\mathbb{S}^{k}$-valued functions of bounded variation defined on $\tdomain$, applying Proposition \ref{prop:entropy-tv-ball} in  Appendix \ref{sec:aux}, 
	we find that  \ref{cond:A1} is also satisfied, and thus  Theorem \ref{THM:SPHERE}
	applies.}
\end{example}

In the above example, all data are located in a subset that has a diameter less than $\pi/2$ and is thus strictly smaller than a hemisphere. If we allow data to be arbitrarily close to the equator, then the constant $c_2$  approaches to zero, and thus the convexity conditions in \ref{cond:A2} might be violated. As pointed out by  a reviewer,  the minimal distance to the equator will play a non-ignorable role, and the convergence rate of Theorem \ref{THM:SPHERE} is expected to change in dependence on  this  minimal distance,  along with changing  constants $C_2$ and $C_3$ in   \ref{cond:A2}. In the extreme case that all data points are located on the equator, the population Fr\'echet mean $\mu$ may not be uniquely defined and thus the total variation regularized estimator might not converge. 
Another extreme case is that the expected Hessian vanishes at the Fr\'echet mean. For this case \cite{Eltzner2019} show that the empirical Fr\'echet mean may still converge to the population Fr\'echet mean,  but at a slower rate. Whether the regularized  estimator proposed here exhibits a similar behavior is of theoretical interest and could be a topic for future research.

\section{Simulation Studies\label{sec:Numerical-Studies}}

We consider three metric spaces,  namely, the
SPD matrix space $\spd$ endowed with the affine-invariant distance in Example \ref{exa:spd} with $m=3$, 
the Wasserstein space $\mathcal{W}_{2}(\real)$ in Example \ref{exa:wass}, and the space of phylogenetic trees in Example \ref{ex:tree}. 
For each of these  metric spaces, two settings are examined. In
the first setting, the underlying mean functions $\mu(t), \, t \in [0,1],$ are locally constant, while in the second setting they smoothly vary with $t\in\tdomain$.  Further details  are given in Table \ref{tab:true-mean-function}. The first setting represents a favorable scenario for total variation regularized Fr\'echet regression, since the estimator is also locally constant, while the second setting is more  challenging. 

For each setting, we investigated two sample sizes,  $n=50$
and $n=150$ for the  design points  $t_{i}=(i-1)/(n-1)$ with  $i=1,\ldots,n$. 
For the 
SPD matrix space, data $Y_i$ were generated as
$Y_{i}=\mu(t_{i})^{1/2}\exp\{\mu(t_{i})^{-1/2}S_{i}\mu(t_{i})^{-1/2}\}\mu(t_{i})^{1/2}$
with $\mathrm{vec}(S_{i})\overset{i.i.d}{\sim}N(0,0.25^{2}I_{6})$, where $\mu(t)$ is as in Table \ref{tab:true-mean-function},  $S_{i}$ is a $3\times3$ symmetric matrix and $\mathrm{vec}(S)$
is its vector representation, i.e., the $6$-dimensional vector obtained
by stacking elements in the lower triangular part of $S$, and 
$I_{6}$ denotes the $6\times6$ identity matrix. 

For the Wasserstein
space, we adopted the method in \citet{Petersen2019} to generate observations $Y_{i}$,
as follows. Let $a_{i}=\expect Z$ and $b_{i}=\{\expect(Z-\expect Z)^{2}\}^{1/2}$
for $Z\sim\mu(t_{i})$, where again the distributions $\mu(t)$ are  as listed in Table \ref{tab:true-mean-function}
for the Wasserstein case. We then first sample $\nu_{i}\sim N(a_{i},1)$
and $\sigma_{i}\sim\mathrm{Gamma}(\alpha,\gamma)$, with  shape parameter $\alpha=0.5b_{i}^{2}$ and rate parameter 
$\gamma=0.5b_{i}$.
Note that $\expect v_{i}=a_{i}$ and $\expect\sigma_{i}=b_{i}$. Then
$Y_{i}$ is obtained by transporting the distribution $N(\nu_{i},\sigma_{i}^{2})$
by a transport map $\mathscr{T}$ that is uniformly sampled from the
collection of maps $\mathscr{T}_{k}(x)=x-\sin(kx)/|k|$ for $k\in\{\pm2,\pm1\}$.
Note that $Y_{i}$ is not a Gaussian distribution due to the transportation.
Nevertheless, one can show that the Fr\'echet mean of $Y_{i}$ is
exactly $\mu(t_{i})$. 

For the case of phylogenetic trees, we generated each $Y_i$ by translating $\mu(t_i)$ along a random geodesic emanating from $\mu(t_i)$ for a random distance that follows the uniform distribution on $[0,0.5]$. This requires identification and computation of geodesics in the tree space $\mathscr{T}_m$ ($m=7$ in our setting), for which we employed the algorithm by \cite{Owen2011}.

The regularization parameter $\lambda$ was chosen by five-fold cross-validation. Specifically, we treated the design points as if they were random, and randomly split the data $\mathcal D\define\{(t_1,Y_1),\ldots,(t_n,Y_n)\}$ into five even partitions $\mathcal D_1,\ldots,\mathcal D_5$. For a given value of $\lambda$, for each $k=1,\ldots,5$, the proposed estimation procedure was applied to $\mathcal D\backslash\mathcal D_k$ to obtain an estimator $\hat\mu_{-k}$. The cross-validation error for the given $\lambda$ was calculated by $\sum_{k=1}^5 \sum_{(t,Y)\in \mathcal D_k} d^2(\hat\mu_{-k}(t),Y)$, and the value of $\lambda$ minimizing  the cross-validation error was selected. The results are based on 100 Monte Carlo runs. 
The estimation quality of $\hat{\mu}$ is quantified by the root integrated squared error (RISE) 
\[
\mathrm{RISE}(\hat{\mu})=\left\{ \int_{\tdomain}d_{\manifold}^{2}(\hat{\mu}(t),\mu(t))\diffop t\right\} ^{1/2}.
\]
The results in Table \ref{tab:mean-function-quality} indicate  that as
sample size grows, the estimation error decreases in  both the favorable setting and the challenging setting. Moreover, we observe that the decay rate of the empirical RISE in the table, defined as the ratio of the RISE with $n=150$ and the RISE with $n=50$, is approximately 0.62. This  seems to agree quite  well with our theory in Section \ref{sec:Asymptotics} that suggests a rate of $(50/150)^{1/3}\approx 0.69$.

\begin{table}\vspace{0cm}
\caption{The mean functions for the metric spaces and settings considered in the simulation study,  where  $I_{3}$
is the $3\times3$ identity matrix,  $N(\nu,\sigma^{2})$ denotes the
Gaussian distribution with mean $\nu$ and variance $\sigma^{2}$, 
$\phi(t)=2(1+e^{-40(t-0.25)})^{-1}$ if $t\in[0,0.5)$ and $\phi(t)=2(1+e^{40(t-0.75)})^{-1}$
if $t\in[0.5,1]$, where the included figures depict the function $\phi$ that is continuous with rapid changes and is used in Setting II, and phylogenetic trees $T_1$, $T_2$ and $T_3$. The length of each edge of these trees is one.}
\label{tab:true-mean-function}
\begin{centering}
\begin{tabular}{|c|c|c|}
\hline 
 & Setting I & Setting II \tabularnewline
\hline 
\hline 
SPD & $\mu(t)=\begin{cases}
I_{3} & t\in[0,\frac{1}{3})\\
2I_{3} & t\in[\frac{1}{3},\frac{2}{3})\\
3I_{3} & t\in[\frac{2}{3},1]
\end{cases}$ &  $\mu(t)=\{1+\phi(t)\}I_{3}$ \tabularnewline
\hline 
Wasserstein & $\mu(t)=\begin{cases}
N(0,1) & t\in[0,\frac{1}{3})\\
N(1,1.5^{2}) & t\in[\frac{1}{3},\frac{2}{3})\\
N(2,2^{2}) & t\in[\frac{2}{3},1]
\end{cases}$ & $\mu(t)=N(\phi(t),\{1+\phi(t)\}^{2})$\tabularnewline
\hline
Tree & $\mu(t)=\begin{cases}
T_1 & t\in[0,\frac{1}{3})\\
T_2 & t\in[\frac{1}{3},\frac{2}{3})\\
T_3 & t\in[\frac{2}{3},1]
\end{cases}$ & $\mu(t)=\llbracket T_1,T_3\rrbracket_{\phi(t)/2}$ \tabularnewline
\hline 
\hline 
\multicolumn{3}{|c|}{\includegraphics[scale=0.75]{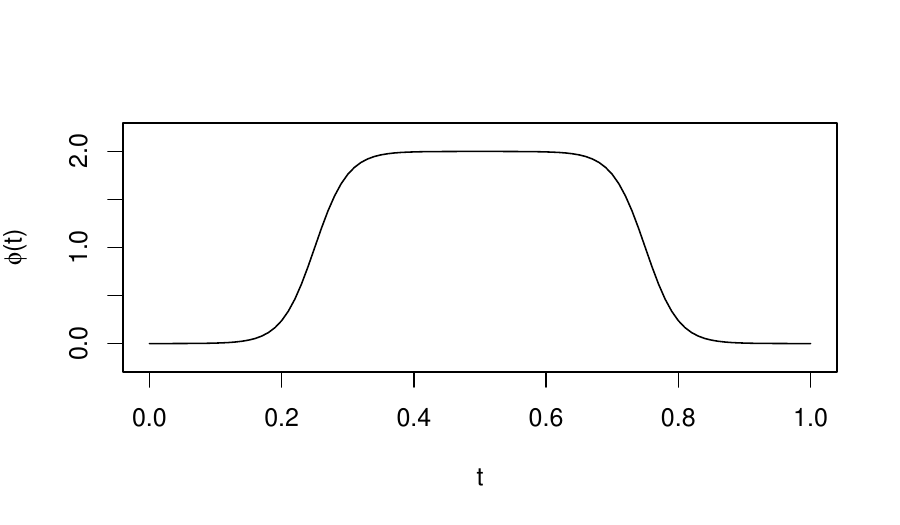}}\tabularnewline
\multicolumn{3}{|c|}{\includegraphics[scale=0.63]{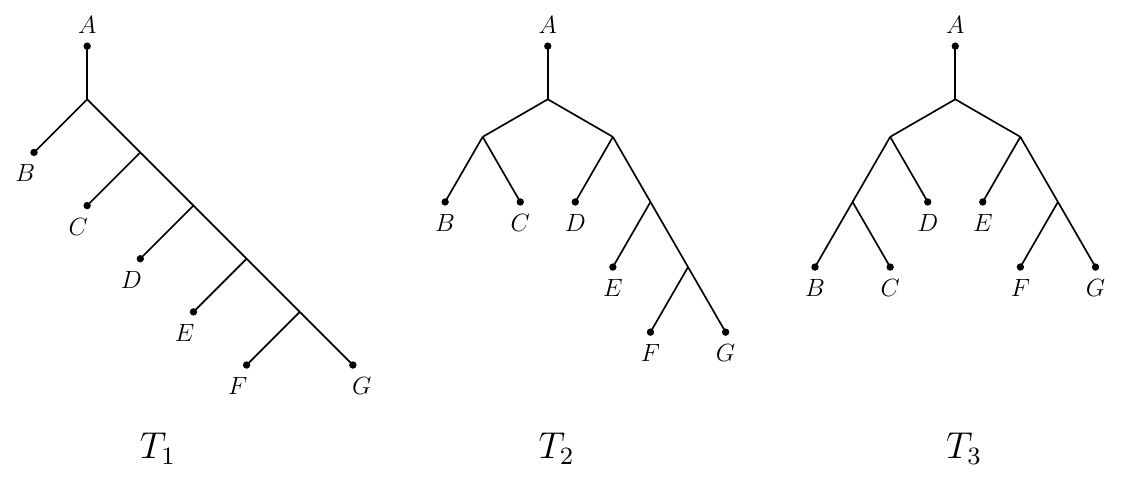}}\tabularnewline

\hline 
\end{tabular}
\par\end{centering}
\end{table}

\begin{table}
\caption{Simulation results for average Root Integrated Squared Error (RISE) of the total variation regularized
estimators for the fitted versus true functions for the two settings considered and random objects corresponding to symmetric positive definite (SPD) matrices, probability distributions with the Wasserstein metric, and phylogenetic trees.  The standard errors based on 100 Monte Carlo replicates
are given in parentheses.}
\label{tab:mean-function-quality}
\begin{centering}
\resizebox{\textwidth}{!}{
\begin{tabular}{|c|c|c|c|c|c|c|}
\hline 
\multirow{2}{*}{Setting} & \multicolumn{2}{c|}{SPD}
& \multicolumn{2}{c|}{Wasserstein} & \multicolumn{2}{c|}{Trees}
\tabularnewline
\cline{2-7} 
 & $n=50$ & $n=150$ & $n=50$ & $n=150$ & $n=50$ & $n=150$\tabularnewline
\hline 
\hline 
I & .210 (.057) & .124 (.042) & .516 (.127) & .321 (.064)& .294 (.116) & .209 (.083)\tabularnewline
\hline 
II & .256 (.054) & .164 (.041) & .604 (.141) & .372 (.073)& .368 (.131) & .235 (.097) \tabularnewline
\hline 
\end{tabular}}
\par\end{centering}
\end{table}

\section{Applications}\label{sec:application}
\subsection{Mortality\label{sec:mortality}}

We applied  the proposed method to analyze the evolution of the distributions
of age-at-death using mortality data from the Human Mortality Database at www.mortality.org.  The database contains yearly mortality for 37
countries, grouped by age from 0 to 110+. Specifically, the data provide a lifetable with a discretization by year, which can be easily converted into a 
histogram of age-at-death, one for each country and calendar year. Starting from these fine-grained histograms, a  simple smoothing step then leads to the density function of age-at-death for a given  
country and calendar year. We  focus on the adult (age 18 or more) mortality
densities of Russia and the calendar years from 1959 to 2014. The time-indexed densities of age-at-death are shown in the form of a heat map in 
Figure \ref{fig:russia-male}(a) for males and for females in Figure \ref{fig:russia-female}(a). The patterns of mortality for males and females are seen to differ substantially.

Applying  the proposed total variation regularized \F regression for distributions as random objects with the Wasserstein distance to these data,  we employ a fine grid on the interval  $[10^{-2.5},10^{-0.1}]$ and  use the aforementioned five-fold cross validation to select the regularization  parameter $\lambda$. The selected values are $\lambda=10^{-1.5}$ and $\lambda=10^{-1.7}$ for males and females,  and the resulting  estimates are shown in Figure \ref{fig:russia-male}(b) and Figure \ref{fig:russia-female}(b), respectively. 

This suggests  that the proposed total variation regularized \F estimator adapts well to the  smoothness of the target function. For example, the female mortality dynamics is seen to be relatively smooth, and the estimator accordingly is also quite smooth. In contrast,  male age-at-death distributions exhibit sharp shifts; the proposed estimator reflects this well and preserves the discontinuities in the mortality dynamics. This  demonstrates desirable flexibility of total variation regularized \F regression, as it  appropriately reflects relatively smooth trajectories, while at the same time preserving edges/boundaries when present. This flexibility  has been documented previously for the Euclidean case \citep{Strong2003}, and is shown here to extend to the much more complex case of metric-space valued data. 

Specifically, a major shift in mortality distributions occurred  around 1992 and is well represented in the estimates for both males and females, with a much larger shift for males. The direction of the shift was towards  increased mortality for both males and females, as the age-at-death distributions moved left, implying increased mortality at younger ages.  A weaker shift that occurred in 2008 is also captured by the estimator  for both males and females, and again is more expressed for males. This latter shift was towards decreased mortality.

These findings pinpoint a period from 1992--2008, during which the turmoil following the collapse of the Soviet Union 1988--1991 appears to have had devastating impacts on mortality. The strong shift in 1992 is relatively easy to explain with social ills such as increased alcoholism and joblessness that followed the collapse of the Soviet Union; it affected males more than females.

\begin{figure}[t]
	\begin{center}
		\includegraphics[scale=0.45]{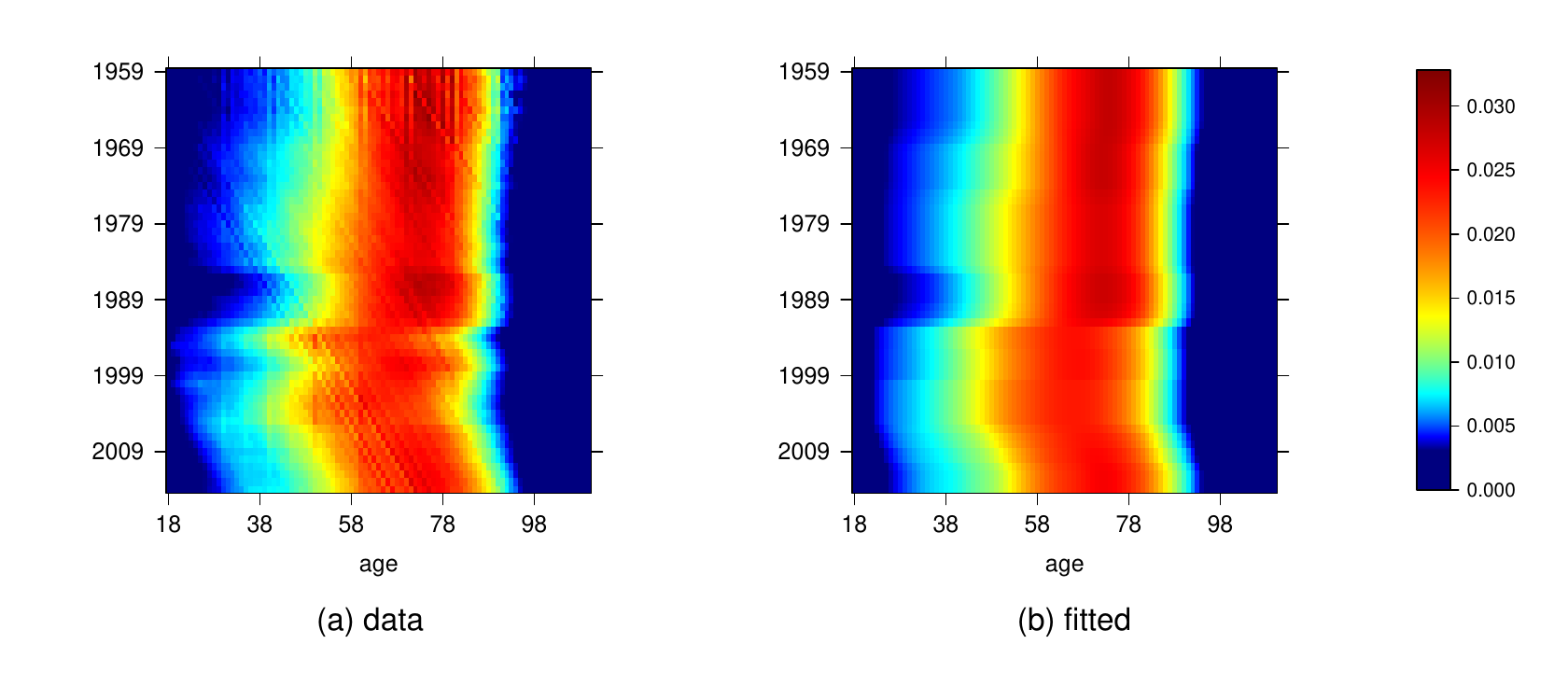}
	\end{center}
	\caption{Total variation regularized Fr\'echet regression for time-indexed mortality distributions of males in Russia, where panel (a) displays  the raw yearly mortality density functions, and panel  (b)  the fitted densities obtained with total variation regularization.} 
	\label{fig:russia-male}
\end{figure}
\begin{figure}[t]
	\begin{center}
		\includegraphics[scale=0.45]{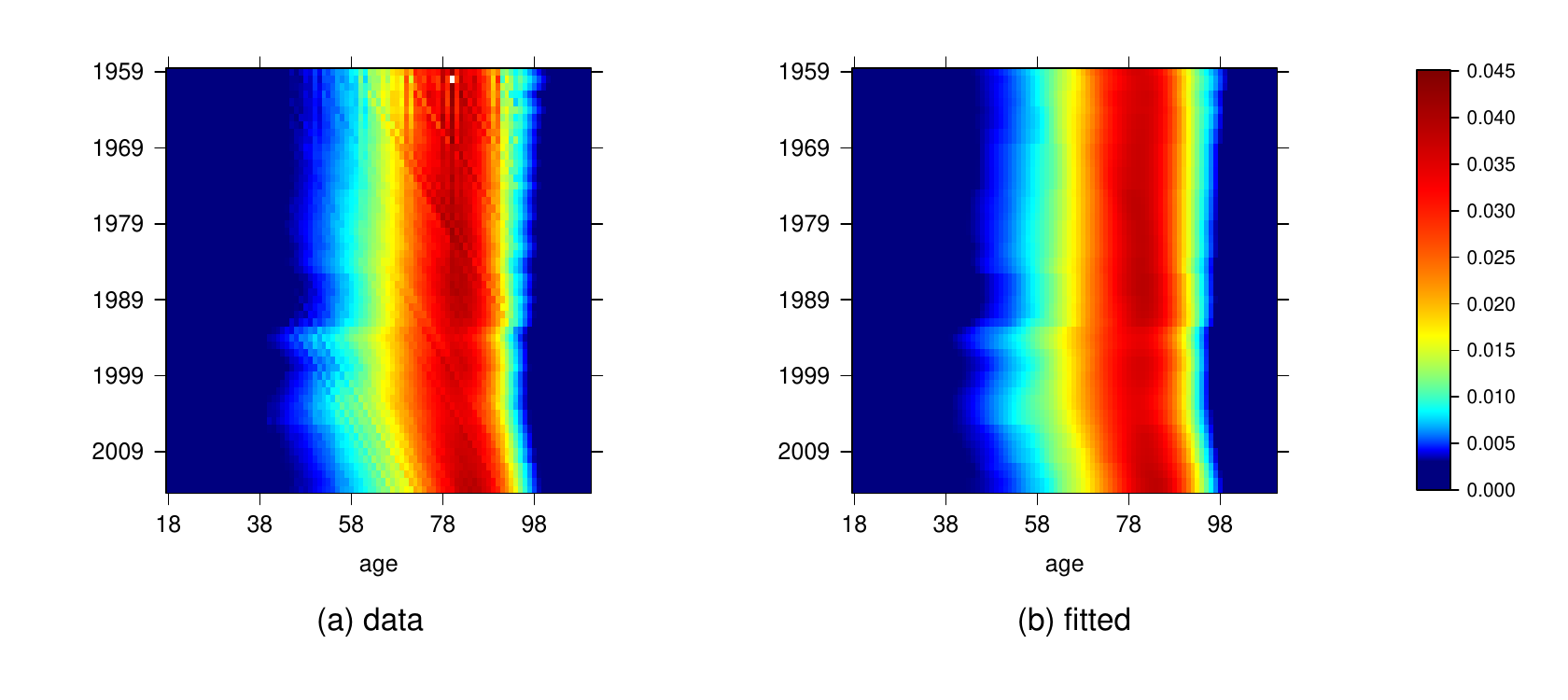}
	\end{center}
	\caption{Total variation regularized Fr\'echet regression for time-indexed  mortality distributions of females in Russia, where panels (a) and (b) are as in Figure~\ref{fig:russia-male}.}
	\label{fig:russia-female}
\end{figure}

\subsection{Functional Connectivity}\label{sec:functional-connectivity}
We applied the proposed total variation regularization method for random objects also to data on functional connectivity in the human brain  from the Human Connectome Project \citep{Essen2013} that were  collected between 2012 and 2015. Out of 970 subjects in the study, for $850$  subjects social cognition task related fMRI data are available.  In this study, each participant was sequentially presented with five short video clips while in a brain scanner, which recorded a fMRI  signal. Each clip showed squares, circles and triangles that either interacted in a certain way or moved randomly.
The fMRI signals were recorded at 274 time points spaced 0.72 seconds apart.  The starting times for the five video clips are approximately  at time points 11, 64, 117, 169 and 222, respectively, with ending times  approximately at time points 39, 92, 144, 197 and 250, respectively, so there are  overall 10  time points where the nature of the visual  input is  changing.  A natural question is then whether  changes in brain connectivity, as quantified by fMRI signals,  are associated with the above time points that indicate changes in visual input. To address this question, we estimated the changes through total variation regularized Fr\'echet regression without using knowledge about the video clip switch times. \linrev{As described in Appendix \ref{sec:app-detail},  we selected 8 brain regions and applied a preprocessing pipeline to obtain the observations $Y_i\in \spd[8]$ at each time point $t_i$, for $i=1,\ldots,n=243$, which are} depicted in  
Figure \ref{fig:fmri32}(a), where for illustration purposes each SPD matrix has been vectorized into an $8(8+1)/2=36$ (taking symmetry into account) dimensional vector represented by a row in the heat map, indicating the relative values of the vector elements. 

This SPD sequence is quite noisy and does not clearly indicate whether the mean brain connectivity changes in accordance with the transition points of the visual input as described above.
Thus, to gain insight whether the pattern of brain connectivity follows the pattern of visual inputs, it is necessary to denoise these data.
Assuming constant brain connectivity while the visual input is constant (video on or off),  this motivates  the fitting of  locally constant functions with a few knots for SPD random objects and thus the application of the proposed total variation regularized \F regression.
This is due to the fact that the proposed estimator $\hat{\mu}$  can be viewed as a locally constant function in time with adaptive knot placement, mapping time into metric space, in our case the space of SPD matrices.   

When applying total variation regularized \F regression,  one has to select the regularization parameter $\lambda$. Generally, we recommend to use the aforementioned cross-validation procedure. However, in  the particular application at hand,  since we may assume that the number of jumps (the discontinuous points of $\mu$) is known to be $J=10$, we  can simply choose the smallest value of $\lambda$ that yields $\hat{J}=10$ jumps of $\hat\mu$. Due to the choice of $P=16>11$ for computing $Y_i$ in Appendix \ref{sec:app-detail}, 
the sequence does not contain sufficient information about the start time point of the first video clip, which is $t=11$. Therefore, we target $J=9$ and choose the smallest value of $\lambda$ that yields $\hat{J}=9$.

Practically, we performed the proposed total variation regularization for the SPD case on the  sequence $Y_i$ for  different choices of the regularization parameter  $\lambda$ on a fine grid within the interval  $[0.01,0.02]$. Panels (b)--(i) in Figure \ref{fig:fmri32} display the resulting estimates by using the affine-invariant distance \citep{Moakher2005,penn:06}; results by using the Log-Euclidean distance \citep{Arsigny2007} are similar. For each panel, the  minimal value of the regularization parameter  $\lambda$  
was chosen so that the number of jump points ranged from 9 (smaller $\lambda$) to 2 (larger $\lambda$),  respectively. From Figure \ref{fig:fmri32}(b), where one has 9 jump points of $\hat\mu$, we find  that the detected jump points closely match the times when the videos clips started and ended, with the exception of time points 11 and 250, which is due to insufficient data between these first and last events and the respective boundaries, and the event at time point 197, which  is split into two jump points, at time points 181 and 202. As $\lambda$ increases, the number of jump points of the estimates decreases. Further discussion can be found in Appendix \ref{sec:app-detail}.

\begin{figure}
\begin{center}
\includegraphics[scale=0.485]{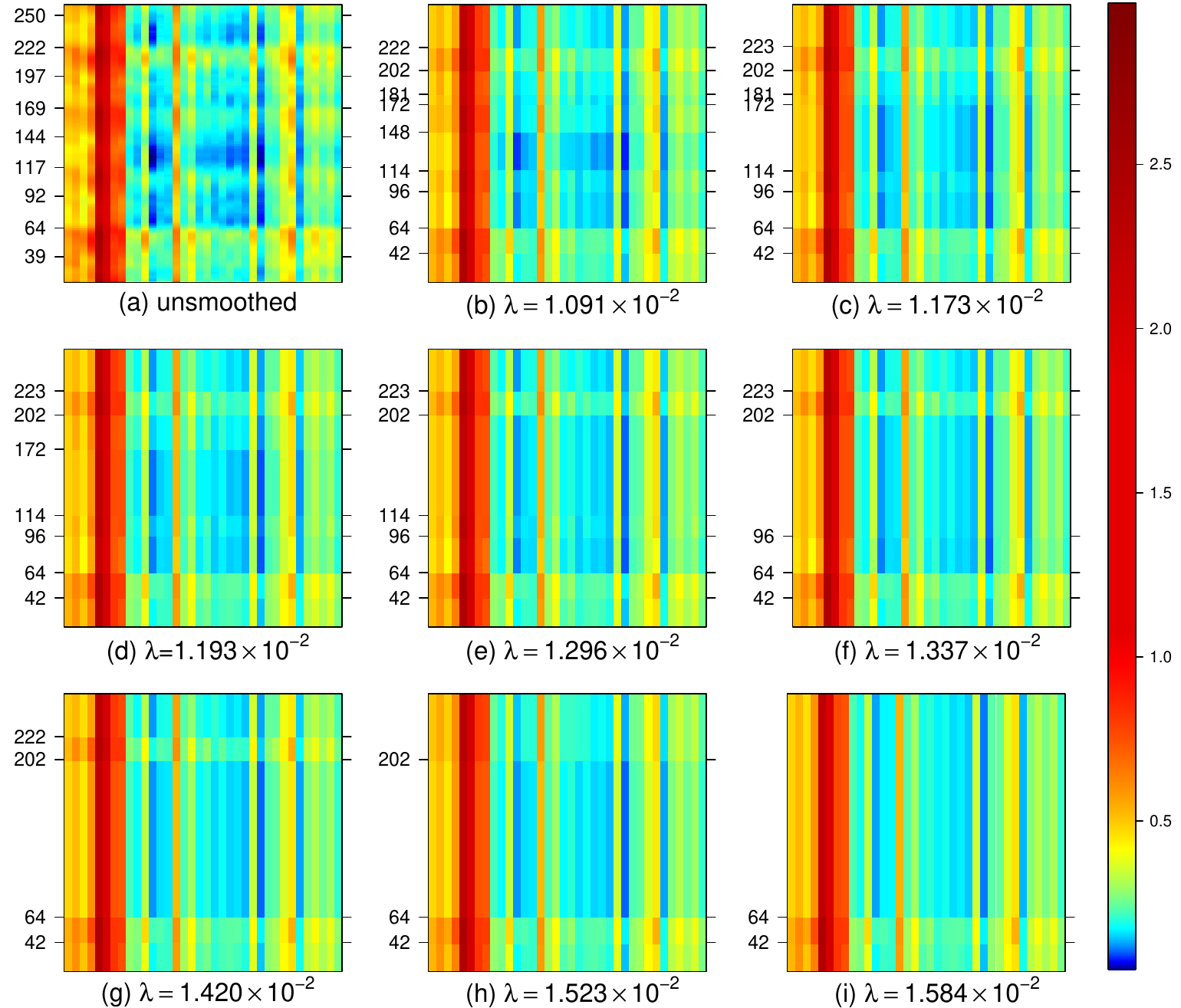}
\end{center}

\caption{Total variation regularized Fr\'echet regression for  dynamic functional connectivity derived from fMRI data. Time is on the vertical axis and the times where visual input changes are explicitly indicated by the tick labels in Panel (a). The lower triangular portions of the SPD covariance matrices of brain connectivity are shown in vectorized form along the horizontal axis. Panel (a) depicts the raw empirical functional connectivity and panels (b)--(i) depict fitted connectivity, obtained by applying the proposed total variation regularized 
\F regression.  From each panel to the next, the  regularization parameter is successively increased such that the number of jump points decreases by one.  For each of the panels (b)--(i), the tick labels on the left side  indicate the locations of the jump points of the fitted step function. In (b) and (c), labels for time points 172 and 181 are overlapping.}

\label{fig:fmri32}
\end{figure}

\section{Concluding remarks}\label{sec:conclusion}
The theoretical developments of the paper are  rooted in convexity of the Hadamard space, which  provides  key ingredients for establishing  the minimax convergence rate for the class of Hadamard spaces. 
The minimax convergence rate is achieved for the one-dimensional Euclidean space, which is a  special case of a Hadamard space. In light of the work \cite{Hotz2013} which shows that the sample Fr\'echet mean converges to its population counterpart at a rate faster than $n^{-1/2}$ in some negatively curved spaces, an interesting future topic is to investigate whether a convergence rate faster than $n^{-1/3}$ is possible for our estimator in some Hadamard spaces of strictly negative curvature.
The convexity also entails an extension to some subspaces of positively curved Alexandrov spaces, where  distance functions are strongly convex over the subspaces as per  condition  \ref{cond:A2:b}. However, a comprehensive treatment for the case of Alexandrov spaces is substantially more challenging, as seen  in Example \ref{exa:sphere} and the related discussion. This requires a theory beyond convexity and thus falls outside of the scope of this paper. 

Extensions to multivariate or manifold-valued domains  are also interesting and nontrivial. 
For multivariate domains, one promising direction is to extend the  
 Hardy--Krause total variation that is utilized by \cite{Fang2021} for multidimensional total variation regularization, since it  carries over to the multidimensional case most of the features of the one-dimensional case, e.g., it is given by the supremum over partitions.
Other interesting and important topics to explore in the future include finite risk bounds and sharp oracle inequalities for the estimated regression function under the setting of Section \ref{sec:Asymptotics}, complementing  the asymptotic theory developed in this paper; sharp bounds  are very challenging  in this setting due to the limited geometric and analytic structure that is available in  general metric spaces.

 Reviewers have  pointed out  that the  entropy condition \ref{cond:H:entropy} is local in nature, in the sense that the constant $K$ might depend on $R$ and it holds only for all $r\leq R$ for an arbitrary but fixed $R>0$ and that  if the entropy condition were global, i.e.,  $\log N(r\delta,\mathscr G_{\manifold}^r (r),d_{n})\leq K\delta^{-1}$ for all $r,\delta>0$,  the proof of Proposition \ref{LEM:KEY} could  be simplified by using a  strategy of \cite{vandeGeer2001}, where the constant  $C$ might  also vary with sample size $n$. It remains however unclear how such a global entropy condition can be verified for the class of general metric-space valued functions of bounded variation. Even for more specific metric spaces such as Riemannian manifolds,
Proposition \ref{prop:entropy-tv-ball} in the supplementary article  \citep{Lin2021}  
suggests that the curvature effect plays an important role in the metric entropy bound. More precise results are left  for future study.

An alternative way to allow  $C$ to vary with $n$ in Proposition \ref{LEM:KEY}, suggested by a reviewer, is to exploit  convexity as  in \cite{Chinot2020}, where one does not require a metric entropy condition. However, \cite{Chinot2020} and the related work \cite{Alquier2019} require  
the concept of Gaussian mean width to characterize complexity of the class of functions under consideration, which is indirectly connected to  metric entropy, e.g., via Sudakov's inequality \citep[Theorem 3.18,][]{Ledoux2011} and Dudley's inequality \citep[Theorem 11.17,][]{Ledoux2011}. Generalization of Gaussian mean width to metric-space valued functions and determining  its precise relation  with metric entropy is another challenging  and interesting topic for future exploration.

\begin{appendix}
	\section{Computational Details}\label{sec:computation}
	To compute
	the total variation regularized estimator defined in (\ref{eq:definition-regularized-estimator}),
	we adopt a simplified version of the cyclic proximal point algorithm
	proposed by \citet{Weinmann2013}. To find the step function estimator according to Proposition \ref{prop:characterization-estimator},
	noting that  $\TV(\hat{\mu})=\sum_{i=1}^{n-1}d(\hat{\mu}(t_{i}),\hat{\mu}(t_{i+1}))$, 
	it is sufficient to compute $\hat{\mu}_{i}\equiv\hat{\mu}(t_{i})$
	for $i=1,\ldots,n$. This is achieved by minimizing the function 
	\[
	\tilde{L}_{\lambda}(p_{1},\ldots,p_{n})=\frac{1}{2}\sum_{i=1}^{n}d^{2}(p_{i},Y_{i})+\frac{n\lambda}{2}\sum_{j=1}^{n-1}d(p_{j},p_{j+1})
	\]
	over the product space $\manifold^{n}$. For ${\bf p}=(p_{1},\ldots,p_{n})$ and 
	$G({\bf p})=\sum_{i=1}^{n}d^{2}(p_{i},Y_{i})$,  the  family of proximal mappings of $G$ is defined by 
	\[
	\mathrm{prox}_{\alpha G}{\bf p}=\underset{{\bf q}\in\manifold^{n}}{\arg\min}\left(\alpha G({\bf q})+\frac{n}{2}d_{n}^{2}({\bf p},{\bf q})\right),
	\]
	where $\alpha>0$ is  a parameter and $d_{n}^{2}({\bf p},{\bf q})=n^{-1}\sum_{i=1}^{n}d^{2}(p_{i},q_{i})$.
	It is easy to check that the $k$th component of $\mathrm{prox}_{\alpha G}{\bf p}$
	is $\llbracket p_{k},Y_{k}\rrbracket_{\theta}$ with $\theta=\alpha(1+\alpha)^{-1}$,
	where we recall that $\llbracket p,q\rrbracket_{\theta}$ denotes the point sitting 
	on the geodesic segment connecting $p$ and $q$ that satisfies  $d(p,\llbracket p,q\rrbracket_{\theta})=\theta d(p,q)$.
	
	For the proximal mappings of the function $H_{j}({\bf p})=d(p_{j},p_{j+1})$, given by \[
	\mathrm{prox}_{\alpha H_{j}}{\bf p}=\underset{{\bf q}\in\manifold^{n}}{\arg\min}\left(\alpha H_{j}({\bf q})+\frac{n}{2}d^{2}_n({\bf p},{\bf q})\right),
	\]
	one finds that  if $k\neq j,j+1$, then the $k$th component of $\mathrm{prox}_{\alpha H_{j}}{\bf p}$
	is equal to $p_{k}$. It is shown in \citet{Weinmann2013} that
	the $j$th component of $\mathrm{prox}_{\alpha H_{j}}{\bf p}$ is given by 
	$\llbracket p_{j},p_{j+1}\rrbracket_{\theta}$, while the  $(j+1)$th
	component is  $\llbracket p_{j+1},p_{j}\rrbracket_{\theta}$,
	where $\theta=\min\{\alpha/d(p_{j},p_{j+1}),1/2\}$ and that the algorithm converges
	to the minimizer of $\tilde{L}_{\lambda}$ 
	for Hadamard spaces.

	The computational
	details are summarized in Algorithm \ref{alg:1}, where the symbol
	$:=$ denotes the assignment or update operator, evaluating 
	the expression on the right hand side and then assigning the value to
	the variable on the left hand side. 
	\begin{algorithm}
		\begin{algorithmic}[1] 
			\Require{$\alpha_1,\alpha_2,\ldots$ such that $\sum_{k=1}^\infty \alpha_k^2<\infty$ and $\sum_{k=1}^\infty \alpha_k=\infty$} 
			\For{$i=1,\ldots,n$}
			\State{$\hat\mu_i := Y_i$}
			\EndFor
			\For{$r=1,2,\ldots$} 
			\For{$i=1,\ldots,n-1$}
			\State{$\theta:=\frac{\alpha_r}{1+\alpha_r}$} 
			\State{$\hat\mu_i:=\llbracket\hat\mu_{i},Y_{i}\rrbracket_{\theta}$} 
			\EndFor
			\For{$j=1,\ldots,n-1$}
			\State{$\theta:=\min\{\alpha_r\lambda n/\{2d(\hat\mu_j,\hat\mu_{j+1})\},1/2\}$}
			\State{$\hat\mu^\prime_j := \llbracket\hat\mu_{j},\hat\mu_{j+1}\rrbracket_{\theta}$ and $\hat\mu^\prime_{j+1} := \llbracket\hat\mu_{j+1},\hat\mu_{j}\rrbracket_{\theta}$}
			\State{$\hat\mu_j := \hat\mu_j^\prime$ and $\hat\mu_{j+1}:=\hat\mu_j^\prime$}
			\EndFor
			\EndFor 
			\Ensure{$\hat\mu(t_i)=\hat{\mu}_1,\ldots,\hat\mu(t_n)=\hat\mu$}
		\end{algorithmic}
		
		\caption{Cyclic Proximal Point Algorithm for Total Variation Regularized Fr\'echet
			Regression}
		\label{alg:1}
	\end{algorithm}

	\section{Key Steps and Geometry in the Proofs}\label{sec:o-lem-key}
	
	To prove Proposition~\ref{LEM:KEY}, we develop novel geometric arguments that make it feasible to extend arguments of  \cite{Mammen1997} to metric-space valued random objects. We first outline how  key arguments used in  \cite{Mammen1997} (rephrased in our context in terms of  language and notations) can be modified to connect  them to the core  ideas of our geometric constructions, and then provide a detailed  proof of the Proposition in the supplementary article \citep{Lin2021}. In this section, the referenced lemmas and equations with labels prefixed by ``S'' are described in \cite{Lin2021}.
	
	There are three key steps in the proofs of Theorem 9 and 10 of \cite{Mammen1997} that were deployed to study  the total variation regularized regression for the traditional situation where  $\manifold=\real$. Once these steps have been  identified and established, the rest of  the proof of \cite{Mammen1997} is  standard. However, these key steps were geared to the linear structure and analytic properties of $\real$, and there is no possibility to modify them for situations without Euclidean structure. To provide versions for  general Hadamard spaces is a serious challenge that we tackle in this paper.  To overcome the technical hurdles, we need to leverage the convexity of Hadamard spaces to obtain geometric versions of these key steps, as follows.
	
	The first key ingredient is the decomposition of $\mu$ into two orthogonal parts by projecting into a space of polynomials and its orthogonal complement. These two parts are handled separately. The complement part is uniformly bounded whenever $\TV(\mu)<C$. The problem is then transformed into estimating a uniformly bounded $\real$-valued function $\mu$ (here we reuse the symbol $\mu$ to conform to the notation  used in our proof) with $\TV(\mu)<C$ via total variation regularization. For  Hadamard spaces, such projections and the space of polynomials do not exist. To circumvent the difficulty, we introduce the concept of \emph{center} of an $\manifold$-valued function $g$, which can be characterized as a  \F integral \cp{mull:16:2,mull:19:1} 
	and is defined to be the minimizer of the function $F_g(p)=\int_{\tdomain} d^2(p,g(t))\diffop t$ over $\manifold$, if $F_g(p)<\infty$ for some $p\in\manifold$. Its discrete version, the center of $g$ at $t_1,\ldots,t_n$, is the minimizer of $F_{g,n}(p)=n^{-1}\sum_{i=1}^n d^2(p,g(t_i))$. Instead of projection, we show that the centers of $\hat\mu$ and $\mu$ at $t_1,\ldots,t_n$ are close to each other  in Lemma \ref{lem:mean-of-mean}. Consequently, we can restrict our focus on functions whose center is close to the center of $\mu$. This makes it possible to  bypass the decomposition of $\mu$ and $\hat\mu$. Note that $\mu(t_1),\ldots,\mu(t_n)$ themselves are the \emph{centers} of the observed data. The center of $\mu$ is then the center of these centers.
	
	A second key ingredient in \cite{Mammen1997} is the inequality $d_n^2(\mu,\hat{\mu})\leq \lambda \{\TV(\mu)-\TV(\hat \mu)\} + 2n^{-1}\sum_{i=1}^n\varepsilon_i\{\hat{\mu}(t_i)-\mu(t_i)\}$, where $\varepsilon_i=Y_i-\mu(t_i)$. In Hadamard spaces, neither the $\varepsilon_i$ nor the differences $\hat{\mu}(t_i)-\mu(t_i)$ or  the products $\varepsilon_i\{\hat{\mu}(t_i)-\mu(t_i)\}$ exist, as these notions are all intimately tied to an underlying Euclidean structure that is not present in metric spaces. To address this challenge, we first use the convexity condition \ref{lem:key:2} in Proposition~\ref{LEM:KEY} to obtain a similar inequality. \linrev{A key step is then to  replace the 
		products $\varepsilon_i\{\hat{\mu}(t_i)-\mu(t_i)\}$ with  $d(\mu(t_i),\hat\mu(t_i))d(\mu(t_i),Y_i)\cos\angle_{\mu(t_i)}(Y_i,\hat\mu(t_i))$, which we refer to as  Alexandrov inner product in this paper}, and to  replace the assumption of zero mean errors  $\mathbb E \varepsilon_i=0$ with the characterization of Fr\'echet means in \eqref{eq:frechet-mean-inequality}. These concepts have not been studied previously to the knowledge of the authors and are likely of more general interest.
	
	The third key ingredient is the observation that  the function $\hat \mu-\mu$, after being scaled by $\TV(\hat \mu)+C$, has total variation bounded by a constant, i.e., $\TV((\hat \mu-\mu)/(\TV(\hat \mu)+C))\leq 1$. This eventually enables one to use Lemma 3.5 of \cite{VandeGeer1990} for the function $(\hat \mu-\mu)/(\TV(\hat \mu)+C)$ in order to  bound the term $n^{-1}\sum_{i=1}^n\varepsilon_i\{\hat \mu(t_i)-\mu(t_i)\}$ by $d^{1/2}(\hat{\mu},\mu)(\TV(\hat \mu)+C)^{1/2}\Op(n^{-1/2})$. Then the rate of $\hat \mu$ can be derived by a standard argument that combines this with the inequality obtained for the second key ingredient. In our context, it is difficult to find a geometric counterpart of $\TV((\hat \mu-\mu)/(\TV(\hat \mu)+C))$ as this involves  subtraction and scaling of functions, which are not available in non-Euclidean spaces.  To overcome this hurdle, we propose the new idea of geodesic interpolation $\tilde g_\theta$ between two functions $\mu$ and $g$, defined by $\tilde g_\theta(t)=\llbracket \mu(t),g(t)\rrbracket_{\theta}$ for $\theta\in[0,1]$. Then the convexity \eqref{eq:pf-lem-C-1} suggests $d(\tilde{g}_\theta(s),\tilde{g}_\theta(t))\leq \theta d(g(s),g(t)) + (1-\theta)d(\mu(s),\mu(t))$ and further  $\TV(\tilde{g}_\theta)\leq \theta\TV(g)+(1-\theta)\TV(\mu)$. In other words, the total variation of the interpolated function $\tilde{g}_\theta$ is bounded by the convex combination of the total variations of $\mu$ and $g$. If we set $\theta=C/\{\TV(g)+C\}$, then $\TV(\tilde{g}_\theta)\leq 2C$ when $\TV(\mu)\leq C$. In particular, this  interpolation preserves the closeness of the centers, i.e., according to Lemma \ref{lem:F-F}, if the center of $g$ is close to $\mu$, then the center of $\tilde{g}_\theta$ is also close to $\mu$. Thus the  interpolation simultaneously mimics the subtraction and scaling of $\mathbb R$-valued functions. This is again a general principle that we expect to be useful  for other investigations where one requires a metric-space counterpart of a standardization procedure that involves function subtraction and scaling.
	
	We note that in order to establish the closeness of the centers in Lemma \ref{lem:mean-of-mean}, we first establish a sub-optimal rate for $\hat{\mu}$ in Lemma \ref{lem:key:rate-with-logn} using last two ideas in the above described key  ingredients. This is made possible by Lemma \ref{lem:max-hat-mu}, where  we use the sub-Gaussianity condition \ref{cond:H:subG} and convexity of the Hadamard space to show that, with probability tending to one, the image of $\hat\mu$ is encompassed by a ball centered at the center of $\mu$ with radius of the order $\log n$. 
	As the proof of Proposition~\ref{LEM:KEY} depends on Lemma \ref{lem:key:rate-with-logn} and  several other proofs  are similar, to avoid repetition, we provide details about the implementation of  the above described ideas mainly  in the proof of Lemma \ref{lem:key:rate-with-logn}, and 
	in the proof of Proposition~\ref{LEM:KEY} those additional  details that are genuinely different from those developed for the proof of Lemma~\ref{lem:key:rate-with-logn}.

\section{Further Details on the Application to Brain Connectivity}\label{sec:app-detail}

\noindent{\it Data Preprocessing}.  
We divided the brain into 68 regions of interest based on the ``Desikan--Killiany'' atlas \citep{DESIKAN2006} and picked eight possible regions that are related to
social skills, i.e., the left and right part of superior temporal, inferior parietal, temporal pole and precuneus \citep{GREEN2015}. The dynamics of functional connectivity for each subject is represented by the changing nature of the 
cross-covariance  between these eight regions, computed by a moving local window that includes $2P$ time points. Specifically, 
denoting by  $V_{ij}$ the vector of the BOLD (blood-oxygen-level dependent) fMRI signals of the $j$th subject at the $i$th time point,  the connectivity at $i=P+1,18,\ldots,274-P+1$ is computed by $$\Sigma_{ij}=\frac{1}{P}\sum_{k=i-P}^{i+P-1}(V_{kj}-\bar{V}_{ij})(V_{kj}-\bar{V}_{ij})^T\quad\text{with}\quad \bar{V}_{ij}=\frac{1}{P}\sum_{k=i-P}^{i+P-1}V_{kj}.$$

In a  last preprocessing step, we aggregated the  information at the same time point across all subjects by computing  $$Y_i=\underset{\Sigma\in\spd[8]}{\arg\min}\frac{1}{850}\sum_{j=1}^{850} d^2(\Sigma,\Sigma_{ij}),$$ where $d$ is the affine-invariant distance \citep{Moakher2005,penn:06} on $\spd[8]$.
	The sequence $Y_1,\ldots,Y_n$ then constituted the observed time-indexed random objects to be analyzed by the proposed regularized Fr\'echet regression.

We set $P=16$ and found that the results were not sensitive to the choice of $P$ within the reasonable range $[12,20]$.  
This led to a sequence of $n=243$ time-indexed  $8\times 8$ covariance (symmetric positive definite, SPD) matrices. For better numerical stability, each matrix was scaled by the constant $10^{-3}$. \vspace{.2cm}

\noindent  {\it Further discussion of the results.} In panel (i) of Figure \ref{fig:fmri32},  there are only two jump points left,   at time points 42 and 64. This suggests that changes in the fMRI signal caused by early events  are more pervasive than those at later events, which is also in line with the fact that the video transition at time point 197 gave rise to two estimated jump points, located slightly before and after.  These findings might be due to a stronger brain reaction to the stimulus  when the video clip is presented early on in the recording sequence, with  subsequent attenuation. 

This example demonstrates that changing the penalty can be used as a tool to determine a hierarchy of jump points with the more pronounced jump points persisting even when large penalties are applied. Remarkably, the location of the estimated jump points is hardly affected by the size of the penalty in this example.

\section{Auxiliary Results}\label{sec:aux}
\begin{prop}\label{prop:entropy}
	{Let $(\mathcal X,d)$ be a metric space that has a finite diameter and satisfies $\sup_{x\in \mathcal X} \log N(\epsilon\delta,B_x(\delta),d)\leq K\epsilon^{-\alpha}$ for constants $\alpha,K>0$ and for all $\epsilon,\delta>0$, where $B_x(\delta)$ denotes the ball in $\mathcal X$ centered at $x$ and with radius $\delta$.
		For a collection  $\mathscr{B}(L)$ of  Lipschitz continuous $\mathcal X$-valued functions defined on $\tdomain$ with a common Lipschitz constant $L<\infty$, it holds that  
		$$\log N(\delta,\mathscr B(L),d_n)\leq  6^\alpha K(2L\delta^{-1}+1)+ 4^\alpha KR^\alpha\delta^{-\alpha},$$ where $R$ denotes the diameter of $\mathcal X$.}
\end{prop}

\begin{prop}\label{prop:entropy-tv-ball}
	{Let $(\manifold,d)$ be a connected smooth Riemannian manifold, and $\Omega\subset\manifold$  a closed uniquely geodesic subspace of diameter $R>0$. Suppose that $\mathscr B\equiv\mathscr{B}(p,D_1,D_2)$ is a collection of $\Omega$-valued functions defined on $\tdomain$ such that $\sup_t d(g(t),p)\leq D_1$ and $\TV(g)\leq D_2$ for some $p\in\Omega$ and all $g\in \mathscr{B}$, where  $D_1,D_2$ are constants. Let $\kappa\geq 0$ be a  constant such that the sectional curvature of $\Omega$ falls into the interval $[-\kappa,\kappa]$.  Then for all $D_1\in(0,R]$ and $D_2>0$, 
		$$ \log N(\delta,\mathscr B,d_n) \leq k\{c_0 k^{1/2}(1+c_\kappa R^2)^2D_2\delta^{-1} + \log(D_1(1+c_\kappa R^2)k^{1/2}\delta^{-1})\},$$
		where $k$ is the dimension of $\manifold$, $c_0$ is an absolute constant and $c_\kappa$ is a constant depending only on $\kappa$.} 
\end{prop}	
	
\end{appendix}

\section*{Acknowledgments}
{We extend our sincere thanks to the editor, associate editor and several referees  for their constructive comments that lead to numerous  improvements over a previous version.}
Data were provided in part by the Human Connectome Project, WU-Minn Consortium (PI: David Van Essen and Kamil Ugurbil; 1U54MH091657) funded by the 16 NIH Institutes and Centers that support the NIH Blueprint for Neuroscience Research; and by the McDonnell Center for Systems Neuroscience at Washington University.

\begin{supplement}
	\stitle{Supplement to ``Total Variation Regularized Fr\'echet Regression for Metric-Space Valued Data''}
	\slink[doi]{COMPLETED BY THE TYPESETTER}
	\sdatatype{.pdf}
	\sdescription{We provide proofs for propositions and theorems in Section \ref{sec:Asymptotics} and Appendix \ref{sec:aux}.}
\end{supplement}

\references



\section*{Supplementary material}

# SUPPLEMENT TO "TOTAL VARIATION REGULARIZED FRÉCHET REGRESSION FOR METRIC-SPACE VALUED DATA"

BY ZHENHUA LIN[1,*] AND HANS-GEORG MÜLLER[2,†]

[1]*Department of Statistics and Applied Probability, National University of Singapore, linz@nus.edu.sg*

[2]*Department of Statistics, University of California, Davis, hgmueller@ucdavis.edu*

## S.1. Proof of Proposition 2.

PROOF OF PROPOSITION 2. In what follows, terms $O_P(\cdot)$ are uniform over the family $\mathscr{F}_n$, i.e., by $X_n = O_P(a_n)$ we mean

$$\lim_{D\to\infty} \limsup_{n\to\infty} \sup_{F\in\mathscr{F}_n} \mathbb{P}_F\{|X_n| > Da_n\} = 0.$$

With the same argument as in the proof of Lemma S.5, we have

$$\text{(S.1)} \quad 0 \geq d_n^2(\hat{\mu}, \mu) - \frac{2}{n} \sum_{i=1}^n d(\mu(t_i), \hat{\mu}(t_i)) d(\mu(t_i), Y_i) \cos \angle_{\mu(t_i)}(Y_i, \hat{\mu}(t_i))$$

$$+ \lambda\{\text{TV}(\hat{\mu}) - \text{TV}(\mu)\},$$

and

$$|V_i(g) - \mathbb{E}V_i(g) - \{V_i(h) - \mathbb{E}V_i(h)\}| \leq C_1 M_i d(g_i, h_i),$$

where $V_i$, $M_i$, $g$ and $h$ are defined in the proof of Lemma S.5. Let $\tilde{g}_\theta$ and $Z_n$ be defined as in the proof of Lemma S.5. If we set $\theta = C/\{\text{TV}(g) + C\}$, then $\text{TV}(\tilde{g}_\theta) \leq 2C$ and therefore $\tilde{g}_\theta \in \mathscr{G}_{\mathcal{M}}(2C)$. In addition, if the center of $g$ is in $\bar{B}_\eta(C)$, where $\eta$ is the center of $\mu$ and $\bar{B}_\eta(C)$ denotes the closed ball centered at $\eta$ with radius $C$, then the center of $\tilde{g}_\theta$ is in $\bar{B}_\eta(6C)$ and further the image of $\tilde{g}_\theta$ is encompassed by $\bar{B}_\eta(8C)$ according to Lemma S.4 and the fact that $\text{TV}(\tilde{g}_\theta) \leq 2C$ and $0 \leq \theta \leq 1$. Now using $\tilde{g} \in \mathscr{G}_{\mathcal{M}}^{8C}(2C) \subset \mathscr{G}_{\mathcal{M}}^{8C}(8C)$ in conjunction with (H2) and Lemma 3.5 of van de Geer (1990), we may conclude that $\sqrt{n}Z_n(\tilde{g}_\theta)/d_n(\tilde{g}_\theta, \mu)^{1/2} = O_P(1)$ for all $\tilde{g}_\theta$ such that the center of $g$ is contained in $\bar{B}_\eta(C)$. As $Z_n(\tilde{g}_\theta) = \theta Z_n(g)$, we further have

$$Z_n(g) = n^{-1/2} \theta^{-1/2} d_n^{1/2}(g, \mu) O_P(1) = d_n^{1/2}(g, \mu)(C + \text{TV}(g))^{1/2} O_P(n^{-1/2}).$$

By Lemma S.6, the center of $\hat{\mu}$ falls into $\bar{B}_\eta(C)$ with probability tending to one. Therefore we can take $g = \hat{\mu}$ in the above, and conclude that

$$\text{(S.2)} \quad Z_n(\hat{\mu}) = \{d_n(\hat{\mu}, \mu)\}^{1/2} (C + \hat{\omega})^{1/2} O_P(n^{-1/2}),$$

where $\hat{\omega} = \text{TV}(\hat{\mu})$, or equivalently,

$$\frac{1}{n} \sum_{i=1}^n d(\hat{\mu}(t_i), \mu(t_i)) d(\mu(t_i), Y_i) \cos \angle_{\mu(t_i)}(\hat{\mu}(t_i), Y_i) - \frac{1}{n} \sum_{i=1}^n U_i(\hat{\mu})$$

$$= \{d_n(\hat{\mu}, \mu)\}^{1/2} (C + \hat{\omega})^{1/2} O_P(n^{-1/2}),$$

---

1

where $U_i(\hat{\mu}) \leq 0$ is defined in the proof of Lemma S.5. Together with (S.1) and $\mathrm{TV}(\mu) \leq C$, we have

$$\text{(S.3)} \qquad d_n^2(\hat{\mu}, \mu) \leq \{d_n(\hat{\mu}, \mu)\}^{1/2}(C + \hat{\omega})^{1/2} O_P(n^{-1/2}) + \lambda\{C - \hat{\omega}\}.$$

Below we address the cases $\hat{\omega} > 2C$ and $\hat{\omega} \leq 2C$ separately.

Case I: $\hat{\omega} > 2C$. From (S.3), we find $0 \leq \{d_n(\hat{\mu}, \mu)\}^{1/2}(2\hat{\omega})^{1/2} O_P(n^{-1/2}) - \lambda\hat{\omega}/2$, from which we further deduce that $\hat{\omega} \leq d_n(\hat{\mu}, \mu)\lambda^{-2} O_P(n^{-1})$. Combining with (S.3), this implies that $d_n^2(\hat{\mu}, \mu) \leq d_n(\hat{\mu}, \mu)\lambda^{-1} O_P(n^{-1})$. The choice of $\lambda$ then implies that $d_n(\hat{\mu}, \mu) = O_P(n^{-1/3})$.

Case II: $\hat{\omega} \leq 2C$. We have $d_n^2(\hat{\mu}, \mu) \leq \{d_n(\hat{\mu}, \mu)\}^{1/2} O_P(n^{-1/2}) + O(\lambda)$ from (S.3). Either $\{d_n(\hat{\mu}, \mu)\}^{1/2} O_P(n^{-1/2})$ dominates $O(\lambda)$, in which case we have $d_n^2(\hat{\mu}, \mu) \leq \{d_n(\hat{\mu}, \mu)\}^{1/2} O_P(n^{-1/2})$, which then is equivalent to $d_n^2(\hat{\mu}, \mu) = O_P(n^{-1/3})$, or $O(\lambda)$ dominates $\{d_n(\hat{\mu}, \mu)\}^{1/2} O_P(n^{-1/2})$, in which case we have $d_n(\hat{\mu}, \mu) = O_P(\lambda^{1/2}) = O_P(n^{-1/3})$. □

## S.2. Proof of Theorem 3.

PROOF OF THEOREM 3. By Lemma S.7 the Fréchet mean $\mu(t_i)$ of $Y_i$ exists and is unique. Also, by (A2c) we obtain the following inequality that is analogous to (S.1),

$$0 \geq C_3 d_n^2(\hat{\mu}, \mu) - \frac{2}{n}\sum_{i=1}^{n} d(\mu(t_i), \hat{\mu}(t_i)) d(\mu(t_i), Y_i) \cos \angle_{\mu(t_i)}(Y_i, \hat{\mu}(t_i))$$
$$+ \lambda\{\mathrm{TV}(\hat{\mu}) - \mathrm{TV}(\mu)\}.$$

The assumed upper bound $\kappa$ on curvature implies that $\mathcal{M}$ is also a CAT($\kappa$) space. Then, by the CAT($\kappa$) inequality (Defintion 1.1 of Chapter II.1, Bridson and Häfliger, 1999) and Lemma S.2, one can establish the following inequality in analogy to (S.10),

$$d(\tilde{g}_\theta(s), \tilde{g}_\theta(t)) \leq c_1 \theta d(g(s), g(t)) + c_1(1-\theta) d(\mu(s), \mu(t)),$$

for some universal constant $c_1 > 0$, where $s, t, \theta, g$ and $\tilde{g}_\theta$ are defined as in the proofs of Proposition 2 and Lemma S.5.

By Lemma S.7, $\mathbb{E}\{d(\mu(t_i), Y_i) \cos \angle_{\mu(t_i)}(Y_i, q)\} \leq 0$ for all $q \in \mathcal{C}$, which is analogous to (c) of Proposition 2. By assumption $\mathrm{diam}(\mathcal{M}) \leq R$ and thus Equation (7) of van de Geer (1990) is automatically satisfied. The proof is then completed by arguments that are similar to those given in the proofs of Proposition 2 and Lemma S.5, and hence is omitted. □

## S.3. Proofs of Propositions 3 and 4.

PROOF OF PROPOSITION 3. Without loss of generality, assume $\mathcal{T} = [0, 1]$. Let $\{x_1, \ldots, x_J\}$ be a maximal $\delta/2$-packing of $\mathcal{X}$, i.e., $J$ is the largest integer such that $d(x_j, x_k) > \delta/2$ for all $1 \leq j \neq k \leq J$. For any $x \in \mathcal{X}$, we then have $d(x, x_j) \leq \delta/2$ for some $j$; otherwise, $\{x, x_1, \ldots, x_J\}$ will form a $\delta/2$-packing with $J + 1$ elements, which contradicts with the maximality of $J$. Thus the sets $\{x \in \mathcal{X} : d(x, x_j) \leq \delta/2\}$ for $j = 1, \ldots, J$ form a $\delta/2$-covering of $\mathcal{X}$. Note that $J = M(\delta/2, \mathcal{X}, d) \leq N(\delta/4, \mathcal{X}, d) \leq \exp\{K(4R/\delta)^\alpha\}$ by the entropy assumption on $\mathcal{X}$, where the notation $M(\delta, \mathcal{X}, d)$ denotes the packing number of $\mathcal{X}$, i.e., the largest number of points that can be packed into $\mathcal{X}$ such that the distance between any two of the points is larger than $\delta$.

Let $m_\delta = \lceil 2L/\delta \rceil$ be the smallest positive integer that is not smaller than $2L/\delta$ and set $A_j = [(j-1)/m_\delta, j/m_\delta)$ for $j = 1, \ldots, m_\delta - 1$, $A_{m_\delta} = [(m_\delta - 1)/m_\delta, 1]$ and $a_j = (2j-1)/(2m_\delta)$. Now, for each $f \in \mathscr{B}(L)$, define $\pi_f$ such that $\pi_f(t) = x_k$ for all $t \in A_j$, where $k$ is the smallest

index such that $x_k$ satisfies $d(x_k, f(a_j)) \leq \delta/2$. For each $t_i$, there exists a $b_i \in \{a_1, \ldots, a_{m_\delta}\}$ such that $|t_i - b_i| \leq \delta/(2L)$ and $t_i, b_i \in A_j$ for some $j$. We claim that the set $\mathscr{L} = \{\pi_f : f \in \mathscr{B}(L)\}$ is a $\delta$-covering of $\mathscr{B}(L)$, by observing

$$d_n(f, \pi_f) \leq \left\{\frac{1}{n}\sum_{i=1}^n d^2(f(t_i), f(b_i))\right\}^{1/2} + \left\{\frac{1}{n}\sum_{i=1}^n d^2(f(b_i), \pi_f(b_i))\right\}^{1/2}$$

$$+ \left\{\frac{1}{n}\sum_{i=1}^n d(\pi_f(b_i), \pi_f(t_i))\right\}^{1/2}$$

$$\leq \left\{\frac{1}{n}\sum_{i=1}^n L^2|t_i - b_i|^2\right\}^{1/2} + \left\{\frac{1}{n}\sum_{i=1}^n \frac{\delta^2}{2^2}\right\}^{1/2} \leq \frac{\delta}{2} + \frac{\delta}{2} = \delta.$$

Now we count the number of elements in $\mathscr{L}$. First, we observe that

$$d(\pi_f(a_j), \pi_f(a_{j+1}))$$
$$\leq d(\pi_f(a_j), f(a_j)) + d(f(a_j), f(a_{j+1})) + d(f(a_{j+1}), \pi_f(a_{j+1}))$$
$$\leq \frac{\delta}{2} + \frac{\delta}{2} + \frac{\delta}{2} = \frac{3\delta}{2}.$$

Since $\pi_f(a_{j+1})$ takes values in $\{x_1, \ldots, x_J\}$ and $d(x_k, x_\ell) > \delta/2$ for $k \neq \ell$, $\pi_f(a_{j+1})$ can only take at most $c_3 := \sup_{x \in \mathcal{X}} M(\delta/2, B_x(3\delta/2), d)$ possible values once $\pi_f(a_j)$ is given. Using $M(\delta/2, B_x(3\delta/2), d) \leq N(\delta/4, B_x(3\delta/2), d)$ and by the entropy assumption on $\mathcal{X}$, we deduce that

$$\sup_{x \in \mathcal{X}} N(\delta/4, B_x(3\delta/2), d) = \sup_{x \in \mathcal{X}} N((1/6)(3\delta/2), B_x(3\delta/2), d) \leq \exp\{6^\alpha K\}.$$

This implies that $c_3 \leq \exp\{6^\alpha K\}$. Then, the cardinality of $\mathscr{L}$ is at most $Jc_3^{m_\delta}$, and $\log N(\delta, \mathscr{B}, d_n) \leq m_\delta \log c_3 + \log J \leq 6^\alpha K(2L/\delta + 1) + K(4R/\delta)^\alpha$. $\square$

PROOF OF PROPOSITION 4. As $\Omega$ is uniquely geodesic, the Riemannian logarithmic map $\text{Log}_p y$, which is the inverse of the corresponding Riemannian exponential map, is defined for all $y \in \Omega$. Define $\phi(y) = \text{Log}_p y$ and represent each $\text{Log}_p y = \sum_{j=1}^k a_j(y) b_j$, where $b_1, \ldots, b_k$ form an orthonormal basis of the tangent space at $p$, and $\mathbf{a}(y) := (a_1(y), \ldots, a_k(y))$ is the coefficient vector of $\text{Log}_p y$. In this way, $y$ is identified with its coefficient vector $\mathbf{a}(y) \in \mathbb{R}^k$. In addition, any $g \in \mathscr{B}$ may be viewed as a $\mathbb{R}^k$-valued function, i.e., $\check{g}(t) := (\check{g}_1(t), \ldots, \check{g}_k(t)) := \mathbf{a}(\phi(g(t)))$. Write $\Phi(g) = \check{g}$.

By Lemma 4 of Sun, Flammarion and Fazel (2019),

$$(1 + c_2(\kappa)R^2)^{-1} d(x, y) \leq \|\phi(x) - \phi(y)\|_p \leq (1 + c_3(\kappa)R^2) d(x, y),$$

where $\|\cdot\|_p$ denotes the norm induced by $\langle \cdot, \cdot \rangle_p$ on $T_p \mathcal{M}$ and $c_2(\kappa), c_3(\kappa)$ are constants depending on $\kappa$. To simplify notation, write $A_2 = 1 + c_2(\kappa)R^2$. Then the above inequalities imply that $\text{TV}(g) \leq A_2 \text{TV}(\check{g})$. In addition, the pre-image of an $(A_2^{-1}\delta)$-covering for $\mathscr{G}_{\mathbb{R}^k}^{D_1}(D_2 A_2)$ under the map $\Phi$ can be transformed into a $\delta$-covering of $\mathscr{B}$, since $d_n(g, f) \leq A_2 \check{d}_n(\check{g}, \check{f})$, where $\check{d}_n(\check{g}, \check{f}) := \sqrt{n^{-1}\sum_{i=1}^n \|\check{g}(t_i) - \check{f}(t_i)\|_2^2}$. Therefore,

$$N(\delta, \mathscr{B}, d_n) \leq N(A_2^{-1}\delta, \mathscr{G}_{\mathbb{R}^k}^{D_1}(A_2 D_2), \check{d}_n).$$

It remains to bound $N(A_2^{-1}\delta, \mathscr{G}_{\mathbb{R}^k}^{D_1}(A_2 D_2), \check{d}_n)$.

Due to $\max_{1 \leq j \leq k} \check{d}_n(\check{g}_j, \check{f}_j) \leq \check{d}_n(\check{g}, \check{f}) \leq \sqrt{k} \max_{1 \leq j \leq k} \check{d}_n(\check{g}_j, \check{f}_j)$ and $\max_{1 \leq j \leq k} \text{TV}(\check{g}_j) \leq \text{TV}(\check{g}) \leq k \max_{1 \leq j \leq k} \text{TV}(\check{g}_j)$, if $\mathcal{A}$ is a $(k^{-1/2}\delta)$-covering of $\mathscr{G}_{\mathbb{R}}^{D_1}(A_2 D_2)$, then $\mathcal{A}^k :=$

$\{\prod_{j=1}^{k} E_j : E_j \in \mathcal{A}\}$ can be transformed into a $\delta$-covering of the space $\prod_{j=1}^{k} \mathscr{G}_{\mathbb{R}}^{D_1}(A_2 D_2) \supset \mathscr{G}_{\mathbb{R}^k}^{D_1}(A_2 D_2)$. This shows that

$$N(A_2^{-1}\delta, \mathscr{G}_{\mathbb{R}^k}^{D_1}(A_2 D_2), \check{d}_n) \le \{N(k^{-1/2} A_2^{-1}\delta, \mathscr{G}_{\mathbb{R}}^{D_1}(A_2 D_2), \check{d}_n)\}^k.$$

As in Section 5 of Mammen (1991), it can be shown that for some absolute constant $c_0$ and for all $D_1$, $D$ and $\delta$,

$$\log N(\delta, \mathscr{G}_{\mathbb{R}}^{D_1}(D), \check{d}_n) \le c_0 D \delta^{-1} + \log(D_1 \delta^{-1}).$$

This implies that

$$\begin{aligned}
\log N(\delta, \mathscr{B}, d_n) &\le \log N(A_2^{-1}\delta, \mathscr{G}_{\mathbb{R}^k}^{D_1}(A_2 D_2), \check{d}_n) \\
&\le k \log N(k^{-1/2} A_2^{-1}\delta, \mathscr{G}_{\mathbb{R}}^{D_1}(A_2 D_2), \check{d}_n) \\
&\le k\{c_0 k^{1/2} A_2^2 D_2 \delta^{-1} + \log(D_1 A_2 k^{1/2} \delta^{-1})\},
\end{aligned}$$

which completes the proof. $\square$

### S.4. Technical Lemmas.

LEMMA S.1. *For a Hadamard space $\mathcal{M}$, for all $p, q, r, u \in \mathcal{M}$,*

$$|d(p,q) \cos \angle_p(u,q) - d(p,r) \cos \angle_p(u,r)| < 5 d(q,r).$$

PROOF. If $d(q,r) > d(p,q)/2$, then

$$\begin{aligned}
&|d(p,q) \cos \angle_p(u,q) - d(p,r) \cos \angle_p(u,r)| \\
&\le d(p,q)|\cos \angle_p(u,q)| + d(p,r)|\cos \angle_p(u,r)| \\
&\le d(p,q) + d(p,r) \le d(p,q) + d(p,q) + d(q,r)
\end{aligned}$$

(S.4) $\qquad < 5 d(q,r).$

Now consider the case $d(q,r) \le d(p,q)/2$. We first observe that

(S.5) $\qquad |d(p,q) \cos \angle_p(u,q) - d(p,r) \cos \angle_p(u,r)|$
$$\le |d(p,q) \cos \angle_p(u,q) - d(p,q) \cos \angle_p(u,r)|$$
$$+ |d(p,q) \cos \angle_p(u,r) - d(p,r) \cos \angle_p(u,r)|,$$

where the second term is bounded by

(S.6) $\qquad |d(p,q) \cos \angle_p(u,r) - d(p,r) \cos \angle_p(u,r)| \le |d(p,q) - d(p,r)| \le d(q,r).$

For the first term,

$$\begin{aligned}
&|d(p,q) \cos \angle_p(u,q) - d(p,q) \cos \angle_p(u,r)| \\
&= d(p,q)|\cos \angle_p(u,q) - \cos \angle_p(u,r)| \\
&\le d(p,q)|\angle_p(u,q) - \angle_p(u,r)| \\
&\le d(p,q) \angle_p qr,
\end{aligned}$$

where the first inequality is due to the Lipschitz continuity of $\cos$ and the last inequality is due to the fact that the (Alexandrov) angle $\angle_p$ induces a pseudo-metric on the space of geodesics leaving $p$ and thus satisfies the triangle inequality. According to Theorem 3.2

of Holopainen (2009), $\angle_p(q,r) \leq \bar{\angle}_p(q,r) = \angle_{\bar{p}}(\bar{q},\bar{r})$, where $\triangle(\bar{p},\bar{q},\bar{r})$ is the comparison triangle of $\triangle(p,q,r)$ in $\mathbb{R}^2$. Since $d(\bar{q},\bar{r}) = d(q,r) \leq d(p,q)/2 = d(\bar{p},\bar{q})/2$, we have $0 \leq \angle_{\bar{p}}(\bar{q},\bar{r}) \leq \pi/6$, and further $\sin \angle_{\bar{p}}(\bar{q},\bar{r}) \leq d(\bar{q},\bar{r})/d(\bar{p},\bar{q}) \leq 1/2$. The monotonicity of arcsin and the Taylor expansion of arcsin at 0 then imply that

$$\angle_{\bar{p}}(\bar{q},\bar{r}) \leq \arcsin\{d(\bar{q},\bar{r})/d(\bar{p},\bar{q})\} = \frac{1}{\sqrt{1-\theta^2}} \frac{d(\bar{q},\bar{r})}{d(\bar{p},\bar{q})},$$

where $\theta \in [0, d(\bar{q},\bar{r})/d(\bar{p},\bar{q})] \subset [0, 1/2]$. Thus,

$$d(p,q) \angle_p(q,r) \leq d(\bar{p},\bar{q}) \angle_{\bar{p}}(\bar{q},\bar{r}) \leq \frac{1}{\sqrt{1-\frac{1}{4}}} d(\bar{p},\bar{q}) \frac{d(\bar{q},\bar{r})}{d(\bar{p},\bar{q})}$$

$$= \frac{2}{\sqrt{3}} d(\bar{q},\bar{r}) = \frac{2}{\sqrt{3}} d(q,r).$$

This combined with (S.4)–(S.6) completes the proof. □

LEMMA S.2. *For $\kappa > 0$ and any $R \in (0, D_\kappa/2)$ with $D_\kappa = \pi/\sqrt{\kappa}$, there exists $c > 0$ such that for all $p, q, r \in B_p(R) = \{x \in M_\kappa^2 : d(p,x) < R\}$, one has $\bar{d}_\kappa(\llbracket p,q \rrbracket_\theta, \llbracket p,r \rrbracket_\theta) \leq c\theta \bar{d}_\kappa(q,r)$ for all $0 \leq \theta \leq 1$.*

PROOF. It is sufficient to prove the lemma for $\theta \in (0,1]$. By the definition of $D_\kappa$, we have $R < \pi/(2\sqrt{\kappa})$, or equivalently, $\sqrt{\kappa}R < \pi/2$. Let $c = 1/\cos(\sqrt{\kappa}R) \in [1,\infty)$. We first claim that

(S.7) $$\theta^{-1}\sin(x) \leq c\sin(x/\theta)$$

for $\theta \in (0,1]$ and $x \in [0, \theta\sqrt{\kappa}R)$. This is implied by $c\cos(x/\theta) \geq c\cos(\sqrt{\kappa}R) = 1$, since $(\theta^{-1}\sin(x))' = \theta^{-1}\cos(x) \leq \theta^{-1} = c\theta^{-1}\cos(\sqrt{\kappa}R) \leq c\theta^{-1}\cos(x/\theta) = (c\sin(x/\theta))'$ and $\theta^{-1}\sin(0) = c\sin(0/\theta) = 0$.

Let $\eta$ be the geodesic connecting $q$ and $r$. Using polar coordinates at $p$, we can write $\eta(s) = (\phi(s), \varphi(s))$. As the curve $\chi(s) = (\theta\phi(s), \varphi(s))$ for $s \in [0, \bar{d}_\kappa(q,r)]$ connects $\llbracket p,q \rrbracket_\theta$ and $\llbracket p,r \rrbracket_\theta$, the distance $\bar{d}_\kappa(\llbracket p,q \rrbracket_\theta, \llbracket p,r \rrbracket_\theta)$ is bounded by the length of $\chi$. Note that $\phi(s) < R$ and $\theta\phi(s) < R$, and thus $\sqrt{\kappa}\phi(s) < \sqrt{\kappa}R$ and $\theta\sqrt{\kappa}\phi(s) < \sqrt{\kappa}R < \pi/2$ for all $s$, whence

$$\bar{d}_\kappa(\llbracket p,q \rrbracket_\theta, \llbracket p,r \rrbracket_\theta) \leq |\chi| = \int_0^{\bar{d}_\kappa(q,r)} |\dot{\chi}|(s)\mathrm{d}s$$

$$= \int_0^{\bar{d}_\kappa(q,r)} \sqrt{\theta^2\{\phi'(s)\}^2 + \kappa^{-1}\sin^2(\theta\sqrt{\kappa}\phi(s))\{\varphi'(s)\}^2}\,\mathrm{d}s$$

$$= \theta \int_0^{\bar{d}_\kappa(q,r)} \sqrt{\{\phi'(s)\}^2 + \kappa^{-1}\theta^{-2}\sin^2(\theta\sqrt{\kappa}\phi(s))\{\varphi'(s)\}^2}\,\mathrm{d}s$$

$$\leq \theta \int_0^{\bar{d}_\kappa(q,r)} \sqrt{c^2\{\phi'(s)\}^2 + \kappa^{-1}c^2\sin^2(\sqrt{\kappa}\phi(s))\{\varphi'(s)\}^2}\,\mathrm{d}s$$

$$\leq c\theta \bar{d}_\kappa(q,r),$$

where the second inequality is due to $c \geq 1$ and (S.7) with $x = \theta\sqrt{\kappa}\phi(s) < \theta\sqrt{\kappa}R$. □

LEMMA S.3. *Assume the conditions of Proposition 2. For the center $\eta$ of $\mu$, we have $\max_{1 \leq i \leq n} d(\eta, \hat{\mu}(t_i)) = O_P(\log n)$ uniformly over $\mathscr{F}_n$. Consequently, $d(\hat{\eta}, \eta) = O_P(\log n)$ for the center $\hat{\eta}$ of $\hat{\mu}$.*

PROOF. The center of a function is as defined in Appendix B. First, the assumption that $\mathrm{TV}(\mu) \leq C$ implies that $\mu$ is bounded, in the sense that the image of $\mu$ is encompassed by a ball of radius at most $C$. Combined with the sub-Gaussianity assumption (**H1**), which implies $\max_{1 \leq i \leq n} d(\mu(t_i), Y_i) = O_P(\log n)$, this suggests $\max_{1 \leq i \leq n} d(\eta, Y_i) = O_P(\log n)$. Therefore, if $R > 0$ is the smallest number such that $\{Y_1, \ldots, Y_n\} \subset \bar{B}_\eta(R)$, then $R = O_P(\log n)$, where $\bar{B}_\eta(R)$ denotes the closed ball centered at $\eta$ and with radius $R$. Below we show that $\max_{1 \leq i \leq n} d(\eta, \hat{\mu}(t_i)) \leq R$ and thus establish the lemma.

The argument is by contradiction. Assume $A = \{i : d(\hat{\mu}(t_i), \eta) = \max_{p \in \mathcal{I}_n} d(p, \eta)\} \neq \varnothing$ but $\{\hat{\mu}(t_i) : i \in A\} \cap \bar{B}_\eta(R) = \varnothing$, where $\mathcal{I}_n = \{\hat{\mu}(t_i) : 1 \leq i \leq n\}$. Let $A_0$ be a subset of $A$ such that the numbers in $A_0$ are consecutive, and $A_0$ is maximal in the sense that both $-1 + \min A_0$ and $1 + \max A_0$ are not in $A$. We define a new estimate $\tilde{\mu}$ by initially setting $\tilde{\mu} = \hat{\mu}$, and then setting $\tilde{\mu}(t_i) = [\![\eta, \hat{\mu}(t_i)]\!]_{\theta_i}$ with $\theta_i = R/d(\eta, \hat{\mu}(t_i)) \in (0, 1)$ for all $i \in A_0$.

First, observe $\mathrm{MSE}(\tilde{\mu}) < \mathrm{MSE}(\hat{\mu})$, where $\mathrm{MSE}(g) = n^{-1} \sum_{i=1}^{n} d^2(g(t_i), Y_i)$ for a $\mathcal{M}$-valued function $g$. To see this, fix $i \in A_0$. By assumption, $d(\eta, \hat{\mu}(t_i)) > R \geq d(Y_i, \eta)$. Let $\bar{\triangle}$ be the comparison triangle of $\triangle(\eta, Y_i, \hat{\mu}(t_i))$ in $\mathbb{R}^2$ and $\bar{z}$ the comparison point of $\tilde{\mu}(t_i)$. Then, by the CAT(0) inequality, $d(Y_i, \tilde{\mu}(t_i)) \leq d(\bar{z}, \bar{Y}_i) < d(Y_i, \hat{\mu}(t_i))$, since the angle at $\bar{z}$ in the comparison triangle $\bar{\triangle}(z, \eta, Y_i)$ is less than $\pi/2$, and hence the angle at $\bar{z}$ in the triangle $\bar{\triangle}(Y_i, \hat{\mu}(t_i), \tilde{\mu}(t_i))$ is larger than $\pi/2$.

Second, we observe that $\mathrm{TV}(\tilde{\mu}) < \mathrm{TV}(\hat{\mu})$. To see this, suppose that $A_0 = \{i+1, \ldots, i+k\}$. Applying the same argument as above shows that $d(\tilde{\mu}(t_i), \tilde{\mu}(t_{i+1})) < d(\hat{\mu}(t_i), \hat{\mu}(t_{i+1}))$ (if $i \geq 1$) and $d(\tilde{\mu}(t_{i+k}), \tilde{\mu}(t_{i+k+1})) < d(\hat{\mu}(t_i), \hat{\mu}(t_{i+1}))$ (if $i + k < n$). In addition, the fact that $\theta_j < 1$ for $j \in A_0$ and the convexity imply that $d(\tilde{\mu}(t_j), \tilde{\mu}(t_{j+1})) < d(\hat{\mu}(t_j), \hat{\mu}(t_{j+1}))$ for $j = i+1, \ldots, i+k-1$ (if $k > 1$).

The above two observations together imply that $L_\lambda(\tilde{\mu}) < L_\lambda(\hat{\mu})$, a contradiction with the minimality of $\hat{\mu}$. An argument similar to the above then leads to $d(\eta, \hat{\eta}) \leq R$ and thus $d(\hat{\eta}, \eta) = O_P(\log n)$. □

LEMMA S.4. *For functions $f$ and $g$ defined on $\mathcal{T}$ that take values in a Hadamard space $\mathcal{M}$, define $\tilde{g}_\theta$ by $\tilde{g}_\theta(t) = [\![f(t), g(t)]\!]_\theta$ and let $\eta_\theta$, $\eta_f$ and $\eta_g$ be the centers of $\tilde{g}_\theta$, $f$ and $g$, respectively. Then $d(\eta_\theta, \eta_f) \leq \mathrm{TV}(\tilde{g}_\theta) + (1+\theta)\mathrm{TV}(f) + \theta d(\eta_f, \eta_g) + \theta \mathrm{TV}(g)$.*

PROOF. For any fixed $t \in \mathcal{T}$,
$$d(\eta_\theta, \eta_f) \leq d(\eta_\theta, \tilde{g}_\theta(t)) + d(\tilde{g}_\theta(t), f(t)) + d(f(t), \eta_f).$$
Let $R(f)$ be the diameter of the image of $f$ and $B(f)$ a ball of diameter $R(f)$ that contains the image of $f$. By an argument similar to that of Lemma S.3, one can show that $\eta_f \in B(f)$ and $R(h) \leq \mathrm{TV}(h)$ for any $\mathcal{M}$-valued function $h$. Therefore, $d(\eta_\theta, \tilde{g}_\theta(t)) \leq R(\tilde{g}_\theta) \leq \mathrm{TV}(\tilde{g}_\theta)$ and $d(f(t), \eta_f) \leq R(f) \leq \mathrm{TV}(f)$. Finally we observe
$$d(\tilde{g}_\theta(t), f(t)) = \theta d(f(t), g(t)) \leq \theta d(f(t), \eta_f) + \theta d(\eta_f, \eta_g) + \theta d(\eta_g, g(t))$$
$$\leq \theta \mathrm{TV}(f) + \theta d(\eta_f, \eta_g) + \theta \mathrm{TV}(g),$$
which concludes the proof. □

LEMMA S.5. *Under the conditions of Proposition 2, $d_n(\hat{\mu}, \mu) = O_P\big(\big(\frac{\log n}{n}\big)^{\frac{1}{3}}\big)$ uniformly over $\mathscr{F}_n$.*

PROOF. We first require several inequalities for the distance function. According to Proposition 1.7 of Chapter II.1 of Bridson and Häfliger (1999), condition (a) implies the CAT(0) inequality, which in turn implies (4.2). We now show that (4.2) further implies that

(S.8) $$d([\![p, q]\!]_\theta, [\![x, y]\!]_\theta) \leq \theta d(q, y) + (1-\theta) d(p, x)$$

for all $\theta \in [0,1]$ and all $p,q,x,y \in \mathcal{M}$. By (4.2), we have $d(\llbracket x,q \rrbracket_\theta, \llbracket x,y \rrbracket_\theta) \leq \theta d(q,y)$ and $d(\llbracket p,q \rrbracket_\theta, \llbracket x,q \rrbracket_\theta) = d(\llbracket q,p \rrbracket_{1-\theta}, \llbracket q,x \rrbracket_{1-\theta}) \leq (1-\theta)d(p,x)$. Now (S.8) follows from the triangle inequality $d(\llbracket p,q \rrbracket_\theta, \llbracket x,y \rrbracket_\theta) \leq d(\llbracket p,q \rrbracket_\theta, \llbracket x,q \rrbracket_\theta) + d(\llbracket x,q \rrbracket_\theta, \llbracket x,y \rrbracket_\theta)$.

Next we establish the convergence rate in the lemma. In what follows, the constants only depend on $K,C,\beta,\zeta$, and $O_P(\cdot)$ denotes a uniform bound over the family $\mathscr{F}_n(K,C,\beta,\zeta)$. Also, it is sufficient to consider $n \geq 2$. The first key observation is

$$0 \geq \frac{1}{n}\sum_{i=1}^n d^2(\hat{\mu}(t_i), Y_i) + \lambda \mathrm{TV}(\hat{\mu}) - \frac{1}{n}\sum_{i=1}^n d^2(\mu(t_i), Y_i) - \lambda \mathrm{TV}(\mu)$$

$$\geq \frac{1}{n}\sum_{i=1}^n \{d^2(\mu(t_i), Y_i) + d^2(\mu(t_i), \hat{\mu}(t_i))\}$$

$$- \frac{2}{n}\sum_{i=1}^n d(\mu(t_i), \hat{\mu}(t_i))d(\mu(t_i), Y_i) \cos\angle_{\mu(t_i)}(Y_i, \hat{\mu}(t_i))$$

$$- \frac{1}{n}\sum_{i=1}^n d^2(\mu(t_i), Y_i) + \lambda\{\mathrm{TV}(\hat{\mu}) - \mathrm{TV}(\mu)\}$$

$$\text{(S.9)} \quad = d_n^2(\hat{\mu}, \mu) - \frac{2}{n}\sum_{i=1}^n d(\mu(t_i), \hat{\mu}(t_i))d(\mu(t_i), Y_i) \cos\angle_{\mu(t_i)}(Y_i, \hat{\mu}(t_i))$$

$$+ \lambda\{\mathrm{TV}(\hat{\mu}) - \mathrm{TV}(\mu)\},$$

where the first inequality is due to the optimality of $\hat{\mu}$ for $L_\lambda$ and the second follows from the convexity condition (a).

For $g \in \mathscr{G}_\mathcal{M}$, let $V_i(g) = d(\mu(t_i), g_i)d(\mu(t_i), Y_i) \cos\angle_{\mu(t_i)}(Y_i, g_i)$ and $Z_n(g) = n^{-1}\sum\{V_i(g) - \mathbb{E}V_i(g)\}$, where $g_i$ denotes $g(t_i)$. Then $\mathbb{E}Z_n(g) = 0$ for all $g$. Also for functions $g$ and $h$,

$$|V_i(g) - V_i(h)|$$
$$= d(\mu(t_i), Y_i)|d(\mu(t_i), g_i)\cos\angle_{\mu(t_i)}(Y_i, g_i) - d(\mu(t_i), h_i)\cos\angle_{\mu(t_i)}(Y_i, h_i)|$$
$$\leq C_1 d(\mu(t_i), Y_i)d(g_i, h_i)$$

according to condition (b), where $h_i$ denotes $h(t_i)$. This also implies that $|\mathbb{E}V_i(g) - \mathbb{E}V_i(h)| \leq C_1\{\mathbb{E}d(\mu(t_i), Y_i)\}d(g_i, h_i)$, and further $|V_i(g) - \mathbb{E}V_i(g) - \{V_i(h) - \mathbb{E}V_i(h)\}| \leq C_1 M_i(g_i, h_i)$, where $M_i := d(\mu(t_i), Y_i) + \mathbb{E}d(\mu(t_i), Y_i)$ is also uniformly sub-Gaussian since $d(\mu(t_i), Y_i)$ is uniformly sub-Gaussian.

Define $\tilde{g}_\theta(t) = \llbracket \mu(t), g(t) \rrbracket_\theta$ for $t \in \mathcal{T}$. The curve $\tilde{g}_\theta$ is viewed as a geodesic interpolation between the curves $\mu$ and $g$ with the parameter $\theta$. The next key observation is that the total variation of $\tilde{g}_\theta$ is bounded by the convex combination of the total variation of $\mu$ and $g$: From convexity (S.8),

$$\text{(S.10)} \quad d(\tilde{g}_\theta(s), \tilde{g}_\theta(t)) \leq \theta d(g(s), g(t)) + (1-\theta)d(\mu(s), \mu(t)),$$

and thus $\mathrm{TV}(\tilde{g}_\theta) \leq \theta \mathrm{TV}(g) + (1-\theta)\mathrm{TV}(\mu)$. If we set $\theta = C/\{\mathrm{TV}(g) + C + c_1\log n\}$ for some constant $c_1 > 0$ to be determined later, then $\mathrm{TV}(\tilde{g}_\theta) \leq 2C$ and therefore $\tilde{g}_\theta \in \mathscr{G}_\mathcal{M}(2C)$. Also, according to Lemma S.4 and the value of $\theta$, we can show that, if the image of $g$ is in $\bar{B}_\eta(c_1 \log n)$ and hence the distance between $\eta$ and the center of $g$ is bounded by $c_1 \log n$, then the center of $\tilde{g}_\theta$ is within $\bar{B}_\eta(6C)$ and further the image of $\tilde{g}_\theta$ is contained in $\bar{B}_\eta(8C)$, where $\eta$ is the center of $\mu$. Hence, $\tilde{g}_\theta \in \mathscr{G}_\mathcal{M}^{8C}(2C) \subset \mathscr{G}_\mathcal{M}^{8C}(8C)$. Using this fact in conjunction with (H2) and Lemma 3.5 of van de Geer (1990) leads to $\sqrt{n}Z_n(\tilde{g}_\theta)/d_n(\tilde{g}_\theta, \mu)^{1/2} = O_P(1)$ for all $\tilde{g}_\theta$ such that $g \in \{h : \text{the image of } h \text{ is in } \bar{B}_\eta(c_1 \log n)\}$. The last key observation is that by construction of $\tilde{g}_\theta$, we have $d_n(\tilde{g}_\theta, \mu) = \theta d_n(g, \mu)$ and $Z_n(\tilde{g}_\theta) = \theta Z_n(g)$, which further implies $Z_n(g) = n^{-1/2}\theta^{-1/2}d_n^{1/2}(g,\mu)O_P(1) = d_n^{1/2}(g,\mu)(C + \mathrm{TV}(g) + c_1\log n)^{1/2}O_P(n^{-1/2})$.

Lemma S.3 enables us to take $g = \hat{\mu}$ in the above by setting $c_1$ to be a sufficiently large constant depending on $C, \beta, \zeta$ and the constant in $O_P(\cdot)$ and to satisfy $c_1 \log 2 > 2C$. Then

$$Z_n(\hat{\mu}) = \{d_n(\hat{\mu}, \mu)\}^{1/2}(C + \hat{\omega} + c_1 \log n)^{1/2} O_P(n^{-1/2}),$$

where $\hat{\omega} = \mathrm{TV}(\hat{\mu})$, or equivalently,

(S.11)
$$\frac{1}{n}\sum_{i=1}^n d(\hat{\mu}(t_i), \mu(t_i)) d(\mu(t_i), Y_i) \cos \angle_{\mu(t_i)}(\hat{\mu}(t_i), Y_i) - \frac{1}{n}\sum_{i=1}^n U_i(\hat{\mu})$$
$$= \{d_n(\hat{\mu}, \mu)\}^{1/2}(C + \hat{\omega} + c_1 \log n)^{1/2} O_P(n^{-1/2}),$$

where $U_i(g) = \mathbb{E} V_i(g) \leq 0$ according to Lemma S.7 for all $g$, in particularly, for $g = \hat{\mu}$. This implies that $n^{-1}\sum_{i=1}^n U_i(\hat{\mu}) \leq 0$. Combining this with (S.9), (S.11) and $\mathrm{TV}(\mu) \leq C$, we conclude that

(S.12) $\quad d_n^2(\hat{\mu}, \mu) \leq \{d_n(\hat{\mu}, \mu)\}^{1/2}(C + \hat{\omega} + c_1 \log n)^{1/2} O_P(n^{-1/2}) + \lambda\{C - \hat{\omega}\}.$

Next we handle the cases $\hat{\omega} \geq c_1 \log n$, $c_1 \log n > \hat{\omega} > C$ and $\hat{\omega} \leq C$ separately.

Case I: $\hat{\omega} \geq c_1 \log n$. From (S.12), we have $0 \leq \{d_n(\hat{\mu}, \mu)\}^{1/2} \hat{\omega}^{1/2} O_P(n^{-1/2}) - \lambda\hat{\omega}/2$, whence $\hat{\omega} \leq d_n(\hat{\mu}, \mu) \lambda^{-2} O_P(n^{-1})$. In combination with (S.12) and the choice of $\lambda$, this implies $d_n(\hat{\mu}, \mu) = O_P(n^{-1/3} \log^{1/3} n)$.

Case II: $c_1 \log n > \hat{\omega} > C$. We have $d_n^2(\hat{\mu}, \mu) \leq \{d_n(\hat{\mu}, \mu)\}^{1/2} (\log n)^{1/2} O_P(n^{-1/2})$ based on (S.12), whence with the choice of $\lambda$, $d_n(\hat{\mu}, \mu) = O_P(n^{-1/3} \log^{1/3} n)$.

Case III: $\hat{\omega} \leq C$. We have $d_n^2(\hat{\mu}, \mu) = \{d_n(\hat{\mu}, \mu)\}^{1/2} O_P(n^{-1/2} \log^{1/2} n) + O(\lambda)$ from (S.12). Either the term $O(\lambda)$ dominates $\{d_n(\hat{\mu}, \mu)\}^{1/2} O_P(n^{-1/2} \log^{1/2} n)$, in which case $d_n(\hat{\mu}, \mu) = O_P(\lambda^{1/2}) = O_P(n^{-1/3})$, or the term $\{d_n(\hat{\mu}, \mu)\}^{1/2} O_P(n^{-1/2} \log^{1/2} n)$ dominates $O(\lambda)$, in which case we have $d_n^2(\hat{\mu}, \mu) = \{d_n(\hat{\mu}, \mu)\}^{1/2} O_P(n^{-1/2} \log^{1/2} n)$, whence $d_n^2(\hat{\mu}, \mu) = O_P(n^{-1/3} \log^{1/3} n)$. □

LEMMA S.6. *Under the conditions of Proposition 2, the center $\hat{\eta}$ of $\hat{\mu}$ at $t_1, \ldots, t_n$ converges to the center $\eta$ of $\mu$ at $t_1, \ldots, t_n$ in probability, uniformly over $\mathscr{F}_n$.*

PROOF. Writing $x_i = \mu(t_i)$ and $\hat{x}_i = \hat{\mu}(t_i)$, by condition (a), we have

$$\frac{1}{n}\sum_{i=1}^n d^2(\hat{x}_i, \eta) \geq \frac{1}{n}\sum_{i=1}^n d^2(\hat{x}_i, \hat{\eta}) - \frac{2}{n}\sum_{i=1}^n d(\hat{x}_i, \hat{\eta}) d(\eta, \hat{\eta}) \cos \angle_{\hat{\eta}}(\eta, \hat{x}_i) + d^2(\eta, \hat{\eta})$$

and

$$\frac{1}{n}\sum_{i=1}^n d^2(x_i, \hat{\eta}) \geq \frac{1}{n}\sum_{i=1}^n d^2(x_i, \eta) - \frac{2}{n}\sum_{i=1}^n d(x_i, \eta) d(\eta, \hat{\eta}) \cos \angle_\eta(\hat{\eta}, x_i) + d^2(\eta, \hat{\eta}).$$

We now use Lemma S.7 to deduce that $-n^{-1}\sum_{i=1}^n d(\hat{x}_i, \hat{\eta}) \cos \angle_{\hat{\eta}}(\eta, \hat{x}_i) \geq 0$ and $-n^{-1}\sum_{i=1}^n d(x_i, \eta) \cos \angle_\eta(\hat{\eta}, x_i) \geq 0$. Then, combining the above displayed inequalities leads to

$$d^2(\hat{\eta}, \eta) \leq n^{-1}\sum_{i=1}^n \{d^2(\hat{x}_i, \eta) - d^2(x_i, \eta) + d^2(x_i, \hat{\eta}) - d^2(\hat{x}_i, \hat{\eta})\}$$

$$\leq n^{-1}\sum_{i=1}^n d(\hat{x}_i, x_i) \{d(\hat{x}_i, \eta) + d(x_i, \eta) + d(x_i, \hat{\eta}) + d(\hat{x}_i, \hat{\eta})\}$$

$$\leq \{n^{-1}\sum_{i=1}^n d^2(\hat{x}_i, x_i)\}^{1/2} \Delta,$$

where $\Delta = \{n^{-1} \sum_{i=1}^{n} [d(\hat{x}_i, \eta) + d(x_i, \eta) + d(x_i, \hat{\eta}) + d(\hat{x}_i, \hat{\eta})]^2\}^{1/2}$. Using the triangle inequalities $d(x_i, \hat{\eta}) \le d(x_i, \eta) + d(\eta, \hat{\eta})$, $d(\hat{x}_i, \hat{\eta}) \le d(\hat{x}_i, \eta) + d(\eta, \hat{\eta})$, and $d(\hat{x}_i, \eta) \le d(\hat{x}_i, x_i) + d(x_i, \eta)$, we obtain

$$\Delta \le \{n^{-1} \sum_{i=1}^{n} [2d(\hat{x}_i, x_i) + 4d(x_i, \eta) + 2d(\eta, \hat{\eta})]^2\}^{1/2}$$

$$\le \{4n^{-1} \sum_{i=1}^{n} [3d^2(\hat{x}_i, x_i) + 12d^2(x_i, \eta) + 3d^2(\eta, \hat{\eta})]\}^{1/2}$$

$$= \{12n^{-1} \sum_{i=1}^{n} [d^2(\hat{x}_i, x_i) + 4d^2(x_i, \eta) + d^2(\eta, \hat{\eta})]\}^{1/2}$$

$$\equiv \{\Delta_1 + \Delta_2 + \Delta_3\}^{1/2}.$$

By Lemma S.5, $\Delta_1 = O_P(n^{-1/3} \log^{1/3} n)$. The boundedness of $\mu$ implies that $\Delta_2 = O(1)$. Therefore,

$$d^2(\hat{\eta}, \eta) = O_P(n^{-1/3} \log^{1/3} n) \{O_P(n^{-1/3} \log^{1/3} n) + O(1) + d^2(\hat{\eta}, \eta)\}^{1/2},$$

which implies that $d(\hat{\eta}, \eta) = O_P(n^{-1/6} \log^{1/6} n)$. □

LEMMA S.7. *Suppose $(\mathcal{M}, d)$ is either*

(a) *a Hadamard space, or*
(b) *a geodesically convex closed subset of an Alexandrov space with curvature in the interval $[0, \kappa]$, such that the diameter of $\mathcal{M}$ is less than $\pi/\sqrt{\kappa}$ and the function $h_y(x) = d^2(x, y)$ is strongly convex with a universal positive constant for all $y \in \mathcal{M}$.*

*Then, for a random element $Z \in \mathcal{M}$ such that $\mathbb{E}d^2(q, Z) < \infty$ for some $q \in \mathcal{M}$, its Fréchet mean $\eta$ exists, is unique, and satisfies*

(S.13) $$\mathbb{E}\{d(\eta, Z) \cos \angle_\eta(p, Z)\} \le 0$$

*for all $p \in \mathcal{M}$ and $p \ne \eta$.*

PROOF. The continuity of the squared distance function implies the continuity of the Fréchet function $F$. For a Hadamard space the function $h_y(\cdot) = d(\cdot, y)$ is strongly convex with parameter 1. Therefore, by the assumptions, the function $h_y$ is strongly convex with a universal positive constant on $\mathcal{M}$ for all $y \in \mathcal{M}$. This also implies the strong convexity of $F$ with the same parameter, and thus the Fréchet mean $\eta$ of $Z$ exists and is unique (Proposition 1.7, Sturm, 2003).

The proof of (S.13) requires the concepts of directions and directional derivatives on metric spaces. To introduce directions, we observe that the angle $\angle_q$ induces a pseudo-distance on the space $\mathcal{G}_q$ of geodesics leaving $p$. The relation $\gamma_1 \sim \gamma_2$ if and only if $\angle_q(\gamma_1, \gamma_2) = 0$ defines an equivalence relation on $\mathcal{G}_q$ and the associated quotient space is denoted by $\Sigma_q^* \mathcal{M}$, and its completion, also called the space of directions at $q$, is denoted by $\Sigma_q \mathcal{M}$. The quotient space $([0, +\infty) \times \Sigma_q^* \mathcal{M})/(\{0\} \times \Sigma_q^* \mathcal{M})$ is called the tangent cone of $\mathcal{M}$ at $q$ and denoted by $C_q^* \mathcal{M}$; we use the notation $C_q \mathcal{M}$ for the quotient space $([0, +\infty) \times \Sigma_q \mathcal{M})/(\{0\} \times \Sigma_q \mathcal{M})$, and identify the subspace $\{1\} \times \Sigma_q \mathcal{M}$ of $C_q \mathcal{M}$ with $\Sigma_q \mathcal{M}$. We also write elements $(t, u)$ in $C_q \mathcal{M}$ by $tu$ and define $|tu| = t$.

For a convex function $f$ defined on $\mathcal{M}$, the directional derivative at $q$ along the direction $tu$ is defined by

(S.14) $$(\mathcal{D}_q f)(tu) := \lim_{s \to 0^+} \frac{f \circ \gamma_u(s) - f \circ \gamma_u(0)}{s} t,$$

where $\gamma_u$ is a geodesic emanating from $q$ (i.e., $\gamma(0) = q$) with unit speed and direction $u$. As convex functions have right derivatives at all points, the above limit is guaranteed to exist by the convexity of $f$.

If $q$ is a minimizer of $f$, then $f \circ \gamma_u(s) - f \circ \gamma_u(0) \geq 0$ for all $s \geq 0$ and $u$, and consequently, $(\mathcal{D}_q f)(tu) \geq 0$ for all directions $tu$. Now we apply this observation to the function $F(\cdot) = \mathbb{E} d^2(\cdot, Z)$. According to Theorem 4.5.6 of Burago, Burago and Ivanov (2001), $(\mathcal{D}_\eta h_y)(u) = -\cos \angle_\eta(p, y)$, where $u \in \Sigma_\eta \mathcal{M}$ such that $p$ is on the geodesic emanating from $\eta$ with direction $u$; such a geodesic exists and is unique, as in case (b) the condition on the diameter ensures that $\mathcal{M}$ is uniquely geodesic. With $(\mathcal{D}_\eta h_y^2)(u) = 2d(\eta, y)(\mathcal{D}_\eta h_y)(u)$, this implies that

$$(\mathcal{D}_\eta F)(u) = \mathbb{E}\{(\mathcal{D}_\eta h_Z^2)(u)\} = -2\mathbb{E}\{d(\eta, Z) \cos \angle_\eta(p, Z)\} \geq 0,$$

which completes the proof. □

## REFERENCES

BRIDSON, M. R. and HÄFLIGER, A. (1999). *Metric Spaces of Non-Positive Curvature*. Springer-Verlag.
BURAGO, D., BURAGO, Y. and IVANOV, S. (2001). *A Course in Metric Geometry*. American Mathematical Society, Providence, RI.
HOLOPAINEN, I. (2009). Metric Geometry Technical Report.
MAMMEN, E. (1991). Nonparametric regression under qualitative smoothness assumptions. *The Annals of Statistics* **19** 741–759.
STURM, K.-T. (2003). Probability measures on metric spaces of nonpositive curvature. In *Heat kernels and analysis on manifolds, graphs, and metric spaces (Paris, 2002), vol. 338 of Contemporary Mathematics* 357–390. American Mathematical Society, Providence, RI.
SUN, Y., FLAMMARION, N. and FAZEL, M. (2019). Escaping from saddle points on Riemannian manifolds In *33rd Conference on Neural Information Processing Systems (NeurIPS 2019)*.
VAN DE GEER, S. (1990). Estimating a regression function. *Annals of Statistics* **18** 907–924.



\end{document}